\documentclass[runningheads]{llncs}

\usepackage{bm}
\usepackage[inline]{enumitem}
\usepackage[utf8]{inputenc}
\usepackage{listings}
\usepackage[numbers,sectionbib,sort&compress]{natbib}
\usepackage{subfig}
\usepackage{todonotes}

\usepackage{auxkling}     

\graphicspath{{images/}}

\setlist[enumerate,1]{label=(\alph*)}
\setlist[enumerate,2]{label=(\roman*)}

\newcommand*{\SL}{\operatorname{\mathcal{L}}}

\newcommand*{\e}{{$\epsilon$-transfer}\xspace}
\newcommand*{\es}{{$\epsilon$-transfers}\xspace}

\title{Optimal Speed Scaling with a Solar Cell}
\author{%
	Neal Barcelo\inst{1}\and
	Peter Kling\thanks{%
		Supported  by fellowships of the Postdoc-Programme of the German Academic Exchange Service (DAAD) and the Pacific Institute of Mathematical Sciences (PIMS).
		Work done while at the University of Pittsburgh.
	}\inst{2}\and
	Michael Nugent\inst{1}\and
	Kirk Pruhs\thanks{Supported, in part, by NSF grants CCF-1115575, CNS-1253218, CCF-1421508, CCF-1535755, and an IBM Faculty Award.}\inst{1}
}
\institute{%
	Department of Computer Science, University of Pittsburgh, Pittsburgh, USA \and
	School of Computing Science, Simon Fraser University, Burnaby, Canada
}

\begin{document}
\maketitle
\begin{abstract}
We consider the setting of a sensor that consists of a speed-scalable processor, a battery, and a solar cell that harvests energy from its environment at a time-invariant recharge rate.
The processor must process a collection of jobs of various sizes.
Jobs arrive at different times and have different deadlines.
The objective is to minimize the \emph{recharge rate}, which is the rate at which the device has to harvest energy in order to feasibly schedule all jobs.
The main result is a polynomial-time combinatorial algorithm for processors with a natural set of discrete speed/power pairs.
\end{abstract}


\section{Introduction}\label{intro}
Most of the algorithmic literature on scheduling devices to manage energy assume the objective of minimizing the total energy usage.
This is an appropriate objective if the amount of available energy is bounded, say by the capacity of a battery.
However, many devices (most notably sensors in hazardous environments) contain energy harvesting technologies.
Solar cells are probably the most common example, but some sensors also harvest energy from ambient vibrations~\cite{Remick:2016aa,Stephen:2006aa} or electromagnetic radiation~\cite{Vullers:2009aa} (e.g., from communication technologies such as television transmitters).
To get a rough feeling for the involved scales (see also~\cite{Vullers:2009aa}), note that batteries can store on the order of a joule of energy per cubic millimeter, while solar cells provide several hundred microwatt per square millimeter in bright sunlight, and both vibrations and ambient radiation technologies provide on the order of nanowatt per cubic millimeter.
Compared to non-harvesting technologies, the algorithmic challenge is to cope with a more dynamic setting, where the difference between \emph{total available} and \emph{total used} energy is non-monotonic (cf.~Figure~\ref{fig:intropict}).

\begin{figure}
\centering
\begin{minipage}{0.45\textwidth}
\includegraphics[width=\textwidth,page=1]{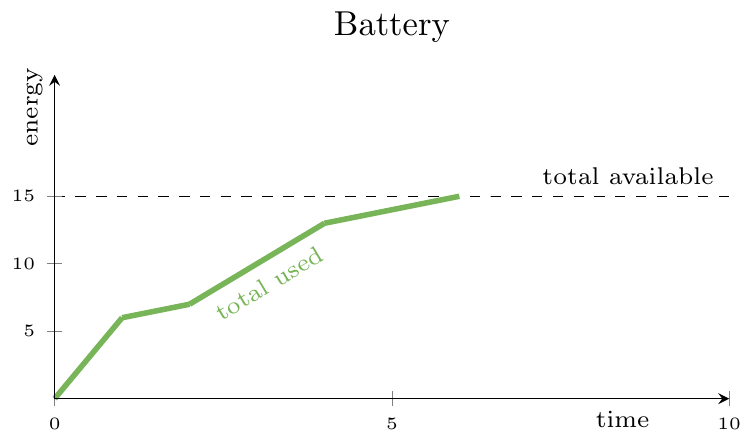}
\end{minipage}
\hfill
\begin{minipage}{0.45\textwidth}
\includegraphics[width=\textwidth,page=2]{intro_picture}
\end{minipage}
\caption{%
	The total available (dashed) and total used (solid) energy for a battery and for a solar cell with a battery.
	The solid line cannot cross the dashed line; when both lines meet, the battery is depleted.
	Depletion is permanent for a battery and temporarily for a solar cell with battery.
	Note that the difference between total available and total used energy is non-monotonic for solar cells.
}
\label{fig:intropict}
\end{figure}

\paragraph{Problem \& Model in a Nutshell}
The goal of this research is to use an algorithmic lens to investigate how the addition of energy harvesting technologies affects the complexity of scheduling such devices.
As a test case, we consider the first (and most investigated) problem on energy-aware scheduling due to \citet{Yao:1995}.
There, the authors assumed that
\begin{enumerate*}
\item the processor is speed-scalable;
\item the power used is the square of the speed; and
\item each job has a certain size, an earliest (release-) time at which it can be run, and a deadline by which it must be finished.
\end{enumerate*}
Their objective was to minimize the total energy used by the processor when finishing all jobs.
We modify these assumptions as follows:
\begin{enumerate}
\item The device has a speed-scalable processor with a \emph{finite} number of speeds $s_1<\dots<s_k$, each associated with a power consumption rate $P_1<\dots<P_k$.
\item The device harvests energy at a time-invariant recharge rate $R>0$ (like a solar-cell in bright sunlight).
\item The device has a battery (initially empty) to store harvested energy.
	To concentrate on the energy harvesting aspect, we assume that the battery's capacity isn't a limiting factor.
\end{enumerate}
The objective becomes to find the minimal necessary recharge rate to finish all jobs between their release time and deadline.

As is the case with the (discrete) variant of~\cite{Yao:1995}, our solar cell problem can be written as a linear program.
Thus, in principle it is solvable in polynomial time by standard mathematical programming methods (e.g., the Ellipsoid method).
However,~\cite{Yao:1995} showed that the total energy minimization problem is algorithmically much easier than linear programming by giving a simple, combinatorial greedy algorithm.
In the same spirit, we study whether the solar cell version allows for a similarly simple, purely combinatorial algorithm.

\paragraph{Results in a Nutshell}
Our main result is a polynomial-time combinatorial algorithm for well-separated processor speeds.
Well-separation is a technical but natural requirement to ease the analysis.
It ensures that the speed/power cover a good efficiency spectrum, as explained below.
Let $\Delta_i\coloneqq\smash{\frac{P_{i}-P_{i-1}}{s_i-s_{i-1}}}$.
The speeds are \emph{well-separated} if there is a constant $c>1$ such that $\Delta_{i+1}=c\cdot\Delta_{i}$ for all $i$.
To understand this condition, note that there is a strong convex relationship between the speed and power in CMOS-based processors~\cite{Brooks:2000}, typically modelled as $\text{Power}=\text{Speed}^{\alpha}$ for some constant $\alpha>1$~\cite{Yao:1995}.
Thus, lower speeds give significantly better energy efficiency.
A chip designer aims to choose discrete speeds (from the continuous range of options) that are well-separated in terms of performance and energy efficiency.
A natural choice is to grow speeds exponentially (i.e., $s_{i+1}=c'\cdot s_{i}$ for a suitable $c'>1$).
With $P_i=s_i^{\alpha}$, we get that speeds are well-separated with the constant $c\coloneqq c'^{\alpha-1}$.

Our algorithm can be viewed as a homotopic optimization algorithm that maintains an energy optimal schedule while the recharge rate is continuously decreased.
Similar approaches for other speed scaling problems have been used in~\cite{puw08,ColeLNP12,AntoniadisBCKNPS14,Angelopoulos2015}.
The resulting combinatorial algorithm exposes interesting structural properties and relations to be maintained while decreasing the recharge rate and adapting the schedule, not unlike (but much more complex than) the homotopic algorithm from~\cite{AntoniadisBCKNPS14}.
While this allows us to prove a polynomial runtime for our algorithm (see Theorem~\ref{thm:algo_runtime}), the actual bound is quite high and only barely superior to bounds derived by generic convex program solvers.
We believe that this runtime is merely an artifact of our hierarchical analysis approach, which aims at simplifying the (already quite involved) analysis.
However, this might also indicate that other, non-homotopical approaches might be more suitable to tackle this scheduling variant.

\paragraph{Context \& Related Results}
The only other theoretical work (we are aware of) on this solar cell problem is by \citet{Bansal:2009}.
They considered arbitrary (continuous) speeds $s\in\R_{\geq0}$ and power consumption $s^{\alpha}$ (where $\alpha>1$ is a constant).
They showed that the offline problem can be expressed as a convex program.
Thus, using the well-known KKT conditions one can efficiently recognize optimal solutions, and standard methods (e.g., the Ellipsoid Method) efficiently solve this problem to any desired accuracy.
\Citet{Bansal:2009} also proved that the schedule that optimizes the total energy usage is a $2$-approximation for the objective of recharge rate.
Finally, they showed that the online algorithm BKP, which is known to be $\LDAUOmicron{1}$-competitive for total energy usage~\cite{Bansal:2007}, is also $\LDAUOmicron{1}$-competitive with respect to the recharge rate.
So, intuitively, the take-away from~\cite{Bansal:2009} was that schedules that naturally arise when minimizing energy usage are $\LDAUOmicron{1}$ approximations with respect to the recharge rate.
In particular, \citet{Bansal:2009} left as an open question whether there is a simple, combinatorial algorithm for the solar cell problem.

\paragraph{Outline}
Both our algorithm design and analysis are quite involved and require significant understanding of the relation between the recharge rate and optimal schedules.
Thus, we start with an informal overview in the next section.
The formal model description and definitions can be found in Sections~\ref{sec:structural_optimality_via_lp} and~\ref{sec:notation}.
The actual algorithm description is given in Section~\ref{sec:algorithm_description}.
Due to space restrictions, most proofs are left for the appendix.


\section{Approach \& Overview}\label{sec:approach_overview}
In the following, we state the central optimality conditions and give a simple illustrating example.
Afterward, we explain how to improve upon a given schedule via suitable transformations guided by these optimality conditions.
Finally, we explain how our algorithm realizes these transformations in polynomial time.

\paragraph{Optimality Conditions}
As the first step in our algorithm design, we consider the natural linear program for our problem and translate the complementary slackness conditions (which characterize optimal solutions) into structural optimality conditions.
This results in Theorem~\ref{thm:optimality}, which states\footnote{%
	Statements slightly simplified; Section~\ref{sec:structural_optimality_via_lp} gives the full formal conditions.
} that optimal solutions can be characterized as follows:
\begin{enumerate}
\item \emph{Feasibility:}
	All jobs are fully processed between their release times and deadlines and the battery is never depleted.
\item \emph{Local Energy Optimality:}
	The job portions scheduled within each \emph{depletion interval} (time between two moments when the battery is depleted) are scheduled in an energy optimal way.
\item \emph{Speed Level Relation (SLR):}
	Consider job $j$ that runs in two depletion intervals $I$ and $I'$.
	Let the average speed of $j$ in $I$ lie between discrete speeds $s_{a-1}$ and $s_{a}$.
	Similarly, let it lie between $s_{b-1}$ and $s_{b}$ in $I'$.
	The SLR states that the difference $b-a$ is independent of the job $j$.
	In other words, jobs jump roughly the same amount of discrete speed levels between depletion intervals\footnote{%
		Figure~\ref{fig:moving_work} gives an example where the SLR can be observed:
		The orange and light-blue jobs run both in depletion interval $I_3$ and $I_4$.
		The orange job's average speed \enquote{jumps} one discrete speed level from $I_3$ to $I_4$ (from below $s_2$ to above $s_2$).
		Thus, the light-blue job must also jump one discrete speed level (from below $s_3$ to above $s_3$).
	}.
\item \emph{Split Depletion Point (SDP):}
	There is a depletion point (time when the battery is depleted) $\tau>0$ such that no job with deadline $>\tau$ is run before $\tau$.
\end{enumerate}
As above for the SLR, we often consider the \emph{average speed} of a scheduled job during a depletion interval.
Note that one can easily derive an actual, discrete schedule from these average speeds:
If a job $j$ runs at average speed $s\in\intco{s_{a-1},s_{a}}$ during a time interval $I$ of length $\abs{I}$, we can interpolate the average speed with discrete speeds by scheduling $j$ first for $\smash{\frac{s_{a}-s}{s_{a}-s_{a-1}}\cdot\abs{I}}$ time units at speed $s_{a-1}$ and for $\smash{\frac{s-s_{a-1}}{s_{a}-s_{a-1}}\cdot\abs{I}}$ time units at speed $s_{a}$.
Using that speeds are well-separated\footnote{%
	In fact, $\Delta_{i+1}>\Delta_{i}$ is already sufficient.
	Also note that starting with the lower speed is essential: otherwise the battery's energy level might become negative.
}, it follows easily that this is an optimal discrete way to achieve average speed $s$.

\paragraph{A Simple Example}
To build intuition, consider a simple example.
The processor has two discrete speeds $s_1=1$ and $s_2=2$ with power consumption rates $P_1=1$ and $P_2=4$, respectively.
Job $j$ is released at time $0$ with deadline $4$ and work $3$.
Job $j'$ is released at time $1$ with deadline $2$ and work $2$.
The \emph{energy optimal} schedule runs job $j'$ at speed $2$ during the time interval $\intcc{1,2}$ and job $j$ at speed $1$ during the time intervals $\intcc{0,1}$ and $\intcc{2,4}$.
It needs recharge rate $R=2.5$.
There is a depletion point $\tau=2$ and two depletion intervals $I_1=\intco{0,\tau}$ and $I_2=\intco{\tau,\infty}$.
See the left side of Figure~\ref{fig:example} for an illustration.
While this schedule fulfills the first three optimality conditions for \emph{rate optimality}, the SDP condition is violated ($j$ is run both before and after $\tau$).
Thus, while it is energy optimal it is not recharge rate optimal.

Consider what happens if we decrease the recharge rate $R$ by an infinitesimal small amount $\varepsilon$ (i.e., decrease the slope of the dotted line in the left part of Figure~\ref{fig:example}).
This results in a negative energy in the battery at time $\tau$ (the solid line in Figure~\ref{fig:example} \enquote{spikes} through the dotted line at $\tau=2$).
This is not allowed, so we have to decrease the energy used before $\tau$.
To do so, we move some work from a job that is processed on both sides of $\tau$ from $I_1$ to $I_2$ (the violation of the SDP guarantees the existence of such a job).
Continuing to do so allows us to decrease the recharge rate until the SDP holds (see the right side of Figure~\ref{fig:example}).
Thus, the resulting schedule is recharge rate optimal (i.e., needs a solar cell of minimal recharge rate).
Also note that this schedule is no longer energy optimal (the total amount of used energy increased).

\begin{figure}
\centering
\begin{minipage}{0.45\linewidth}
\includegraphics[width=\linewidth,page=1]{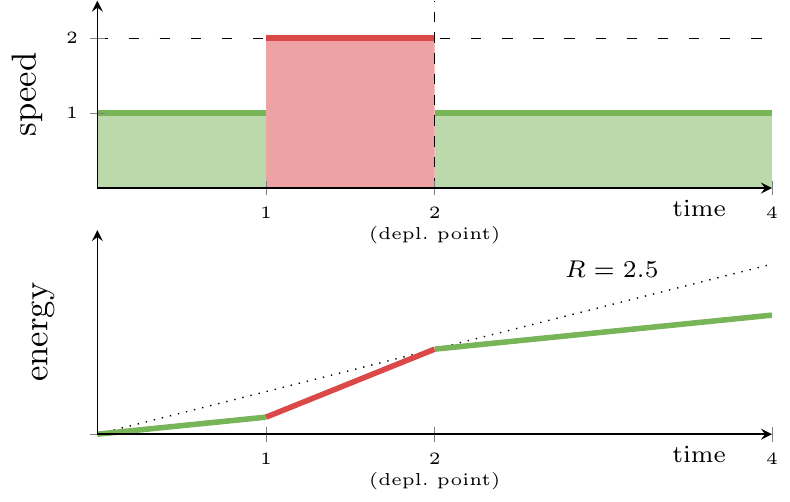}
\end{minipage}
\hfill
\begin{minipage}{0.45\linewidth}
\includegraphics[width=\linewidth,page=2]{saving_energy}
\end{minipage}
\caption{%
	The energy optimal (left) and recharge rate optimal (right) schedules.
	Job speeds are plotted as average speeds (i.e., the green job running at average speed $3/2$ in the depletion interval $\intco{2,4}$ on the right runs in the actual, discrete schedule at speed $1$ during $\intco{2,3}$ and at speed $2$ during $\intco{3,4}$).
}
\label{fig:example}
\end{figure}

\paragraph{Algorithmic Intuition}
Our algorithm extends on the schedule transformation we saw in the simple example above.
We start with an energy optimal schedule $S$ and a trivial bound on the recharge rate $R$ such that the first three optimality conditions hold.
We then lower $R$ while maintaining a schedule satisfying these first three conditions until, additionally, the SDP holds.
Lowering the recharge rate $R$ means we have to move work out of each depletion interval (or we get a negative energy in the battery).
Since we want to maintain the first three optimality conditions, we cannot move work arbitrarily.
To capture all constraints while moving work we employ a \emph{distribution muligraph} $G_D$.
Its vertices are the depletion intervals, and there is a directed edge for each way in which work can be transferred between depletion intervals.
See Figure~\ref{fig:moving_work} for an illustration.
\begin{figure}
\centering
\includegraphics[width=0.618\linewidth]{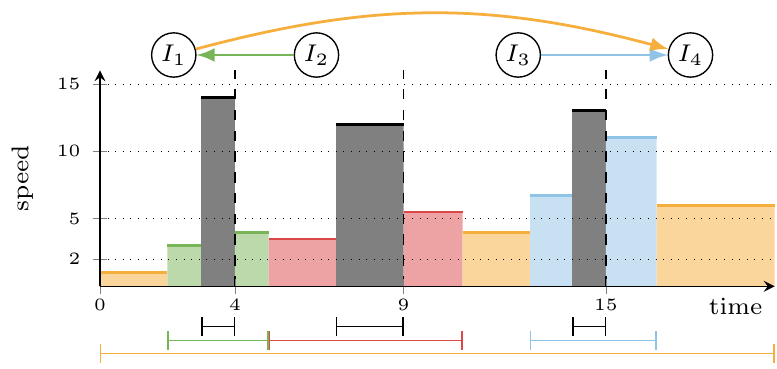}
\caption{%
	Four discrete speeds ($s_1=2,s_2=5,s_3=10,s_4=15$), seven jobs (release/deadlines indicated by the colored bars).
	Four depletion points $\tau\in\set{0,4,9,15}$ form four depletion intervals $I_1=\intco{0,4},I_2=\intco{4,9},I_3=\intco{9,15},I_4=\intco{0,\infty}$.
	A subgraph of the distribution graph $G_D$ is shown above.
}
\label{fig:moving_work}
\end{figure}

The heart of our algorithm is to find a suitable \emph{transfer path} for each depletion interval: a path over which work can be transferred to the rightmost depletion interval (possibly via multiple jobs).
Given such transfer paths, we can move work out of every depletion interval.
While this allows us to make progress, there are three types of events that can occur and must be handled:
\begin{itemize}
\item \emph{Edge Removal Event:} It is no longer possible to transfer work on a particular edge because there is no more work left on the job we were moving.
\item \emph{Depletion Point Appearance Event:} A new depletion point is created.
\item \emph{Speed Level Event:} Further transfer of work would cause a job's average speed in a depletion interval to cross a discrete speed (possibly violating the SLR).
\end{itemize}
In these cases, our algorithm attempts to find a different collection of transfer paths.
If this is not possible, the algorithm tries to update $G_D$ as follows:
\begin{itemize}
\item \emph{Depletion Point Removal Update:} Find a depletion point that can be removed.
	Removing the constraint that the battery is depleted at this point may allow for new ways to transfer work.
\item \emph{Cut Update:}
	Because of the SLR, jobs have to jump the same amount of discrete speed levels between two depletion intervals.
	Thus, all jobs must cross the next discrete speed at the same time.
	A cut update basically signals that all involved jobs reached a suitable discrete speed and can now cross the discrete speed level.
	See Section~\ref{sec:notation} and Definition~\ref{def:slr} for details.
\end{itemize}
As an example, consider what happens when moving work of the light blue job from $I_3$ to $I_4$ in Figure~\ref{fig:moving_work}.
After a while, its average speed in $I_3$ reaches the discrete speed $s_2$ (a speed level event).
The SLR forbids to further decrease this job's speed (it would jump two discrete speed levels, while the orange job jumps only one).
Instead, we start to move work of the orange job from $I_3$ to $I_4$ until it hits the discrete speed $s_1$.
All jobs processed before and after depletion point $\tau=15$ are now at suitable discrete speeds and we can allow both the orange and light blue job to further decrease their speeds (a cut update).

Our correctness proof shows that, if none of these updates is possible, the SDP holds.
A technical difficulty is that events might influence each other, resulting in complex dependencies (which we ignored in the above example).
A lot of the complexity of our algorithm/analysis stems from an urge to avoid these dependencies wherever possible.
However, it seems likely that a more careful study of these dependencies would yield a significant simplification and improvement.

\paragraph{Events \& Updates in Polynomial Time}
Our description above assumes that we move work continuously and stop at the corresponding events.
To implement this in our algorithm, we have to calculate the next event for the current collection of transfer paths and then compute the correct amount of work to move between all involved depletion intervals.
While the involved calculations follow from a simple linear equation system, the main difficulty is to show that the number of events remains polynomial.
To ensure this, our algorithm design facilitates the following hierarchy of invariants:
\begin{itemize}
\item \emph{Cut Invariant:} Cut updates are at the top of the hierarchy.
	Intuitively, this invariant states that job-speeds (or speed levels) tend to increase toward the right (since, as a net effect, work is generally moved to the right).
	This is, for example,  used in Lemma~\ref{lem:runtime_cut_fixing} to prove that there is only a polynomial number of cut updates.
\item \emph{Depletion Point Removal Invariant:} Depletion point updates are at the second level of the hierarchy.
	This invariant states that once a depletion point is removed it will not be added again (until the next cut update).
\item \emph{Speed Level Invariant:} Speed level updates are also at the second level of the hierarchy.
	This invariant states that this event creates a time interval to which no work is added (until the next cut update).
\item \emph{Edge Removal Invariant:} Edge removal events are at the bottom of the hierarchy.
	This invariant states that once work of a job was transferred to an earlier depletion interval (\enquote{to the left}), it will not be transferred to a later one (\enquote{to the right}) until the next cut, depletion point removal, or speed level update.
\end{itemize}
This hierarchy provides a monotone progress measure, but complicates the algorithm/analysis quite a bit.
In particular, we have to deal with two aspects:
\begin{enumerate*}
\item $G$ together with all transfer paths may be exponentially large.
	To handle this, we search for suitable transfer paths on a subgraph $H$ of $G$, containing only the best transfers to move work between any given pair of depletion intervals.
\item We have to define how to select these collections of transfer paths.
	On a high level, the algorithm prefers transfers that move work right to transfers that move work left.
	Between transfers moving work right it prefers shorter transfers, while between transfers moving work left it prefers longer transfers (cf.~Definition~\ref{def:eps:domination}).
\end{enumerate*}


\section{Structural Optimality via Primal-Dual Analysis}\label{sec:structural_optimality_via_lp}
We model our problem as a linear program and use complementary slackness conditions to derive structural properties that are sufficient for optimality.
These structural properties are used in both the design and the analysis of the algorithm.

\subsection{Model}
We consider the problem of scheduling a set of $n$ jobs $J\coloneqq\set{1,2,\dots,n}$ on a single processor that features $k$ different speeds $0<s_1<s_2<\dots<s_k$ and that is equipped with a solar-powered battery.
The battery is attached to a solar cell and recharges at a rate of $R\geq0$.
The power consumption when running at speed $s_i$ is $P_i>0$.
That is, while running at speed $s_i$ work is processed at a rate of $s_i$ and the battery is drained at a rate of $P_i$.
When the processor is idling (not processing any job) we say it runs at speed $s_0\coloneqq0$ and power $P_0\coloneqq0$.

Each job $j\in J$ comes with a \emph{release time} $r_j$, a \emph{deadline} $d_j$, and a processing volume (or work) $p_j$.
For each time $t$, a schedule $S$ must decide which job to process and at what speed.
Preemption is allowed, so that a job may be suspended at any point in time and resumed later on.
We model a schedule $S$ by two functions $S(t)$ (\emph{speed}) and $J(t)$ (\emph{scheduling policy}) that map a time $t \in \mathbb{R}$ to a speed index $S(t)\in\set{0,1,\dots,k}$ and a job $J(t)\in J$.
We say a job $j$ is \emph{active} at time $t$ if $t\in\intco{r_j,d_j}$.
Jobs can only be processed when they are active.
Thus, a \emph{feasible schedule} must ensure that $J^{-1}(j)\subseteq\intco{r_j,d_j}$ holds for all jobs $j$.
Moreover, a feasible schedule must finish all jobs and must ensure that the energy level of the battery never falls below zero.
More formally, we require $\int_{J^{-1}(j)}s_{S(t)}\dif{t}\geq p_j$ for all jobs $j$ and $\int_0^{t_0}P_{S(t)}\dif{t} \leq R \cdot t_0$ for all times $t_0$.
Our objective is to find a feasible schedule that requires the minimum recharge rate.

\subsection{Linear Programming Formulation}
For the following linear programming formulation, we discretize time into equal length time slots $t$.
Without loss of generality, we assume that their length is such that there is a feasible schedule for the optimal recharge rate $R$ that processes at most one job using at most one discrete speed in each single time slot\footnote{%
	The existence of such a schedule follows from standard speed scaling arguments.
	To see this, note that any schedule can be transformed to use earliest deadline first and interpolate an average speed in a depletion interval by at most one speed change between two discrete speeds.
	Thus, the number of job changes and speed changes is finite (depending on $n$) and we merely have to choose the time slots suitably small.
}.
Our linear program uses indicator variables $x_{jit}$ that state whether a given job $j$ is processed at a speed $s_i$ during time slot $t$.
Note that not only does this imply a possible huge number of variables but it is also not trivial to compute the length of the time slots.
Nevertheless, this will not influence the running time of our algorithm, since we merely use the linear program to extract sufficient structural properties of optimal solutions.
Our analysis will also use the fact that we can always further subdivide the given time slots into even smaller slots without changing the optimal schedule.
By rescaling the problem parameters, we can assume that the (final) time slots are of unit length.

With the variables $x_{jit}$ as defined above and a variable $R$ for the recharge rate, the integer linear program (ILP) shown in Figure~\ref{fig:linear_program:primal} corresponds to our scheduling problem.
The first set of constraints ensure that each job is finished during its release-deadline interval, while the second set of constraints ensures that the battery's energy level does not fall below zero.
The final set of constraints ensures that the processor runs at a constant speed and processes at most one job in each time slot.

\begin{figure}
\subfloat[ILP for our scheduling problem.]{%
	\begin{minipage}[b][4.9cm][t]{0.44\linewidth}
	\begin{align*}
	\min         &  \mathrlap{\quad R}\\
	\text{s.t.}  &&   \sum_{t\in\intco{r_j,d_j}}\sum_ix_{jit} \cdot s_i                        &\geq p_j   &   \forall j\\
	             &&   \mathllap{R \cdot t-\sum_{t'\leq t}\sum_{j\in J}\sum_{i=1}^kx_{jit'} \cdot P_i} &\geq 0     &   \forall t\\
	             &&   \sum_{j\in J}\sum_{i=1}^kx_{jit}                                  &\leq 1     &   \forall t\\
	             &&   \mathllap{x_{jit}\in\set{0,1}}                                    &           &   \forall j, i, t
	\end{align*}\label{fig:linear_program:primal}
	\end{minipage}
}
\hfill
\subfloat[Dual program for the ILP's relaxation.]{%
	\begin{minipage}[b][4.9cm][t]{0.49\linewidth}
	\begin{align*}
	\max         &  \mathrlap{\quad \sum_{j\in J}\alpha_j \cdot p_j-\sum_t\gamma_t}\\
	\text{s.t.}  &&   \alpha_j \cdot s_i-\sum_{t'\geq t}\beta_{t'} \cdot P_i-\gamma_t &\leq 0   &   \forall j,i,t\\
	             &&   \sum_t\beta_t \cdot t                                    &\leq 1   &   \\
	             &&   \alpha_j,\beta_t,\gamma_t &\geq 0   &   \forall j, t
	\end{align*}\label{fig:linear_program:dual}
	\end{minipage}
}
\caption{}
\label{fig:linear_program}
\end{figure}

\paragraph{Structural Properties for Optimality}
The complementary slackness constraints for the programs shown in Figure~\ref{fig:linear_program} give us necessary and sufficient properties for the optimality of a pair of feasible primal and dual solutions.
A description of these conditions can be found in Appendix~\ref{sec:structural_optimality_conditions}.
Although these conditions are only necessary and sufficient for optimal solutions of the ILP's \emph{relaxation}, our choice of the time slots ensures that there is an integral optimal solution to the relaxation.
Based on these complementary slackness constraints, we derive some purely combinatorial structural properties (not based on the linear programming formulation) that will guarantee optimality.
To this end, we will consider \emph{speed levels} of jobs in depletion intervals – essentially the discrete speed a job reached in a specific depletion interval – and how they change at depletion points.
In the following, if we speak of a speed $s$ between two discrete speeds (e.g., $s_2<s<s_3$) we implicitly assume $s$ to refer to the average speed in the considered time interval.
\begin{definition}[Speed Level Relation]\label{def:slr}
A schedule $S$ and a recharge rate $R$ obey the \emph{Speed Level Relation} (SLR) if there exist natural numbers $\SL(j,\ell)\in\N$ (\emph{speed levels}) such that
\begin{enumerate}
\item\label{def:slr:a} job $j$ processed at speed $s_{j,\ell}\in\intoo{s_{i-1},s_i}$ in depletion interval $I_{\ell}$ $\Rightarrow$ $\SL(j,\ell)=i$
\item\label{def:slr:b} job $j$ processed at speed $s_{j,\ell}=s_i$ in depletion interval $I_{\ell}$ $\Rightarrow$ $\SL(j,\ell)\in\set{i,i+1}$
\item\label{def:slr:c} jobs $j,j'$ both active in depletion intervals $I_{\ell_1}$ and $I_{\ell_2}$ with $\ell_1<\ell_2$ $\Rightarrow$ $\SL(j,\ell_2)-\SL(j,\ell_1)=\SL(j',\ell_2)-\SL(j',\ell_1)\in\N_0$ (in particular, the speed levels of a job are non-decreasing)
\item\label{def:slr:d} job $j$ processed in $I_l$ $\Rightarrow$ $\SL(j,l)\geq\SL(j',l)$ for all $j'$ active in $I_{l,j}=I_l\cap\intco{r_j,d_j}$
\end{enumerate}
\end{definition}
Intuitively, the SLR states that jobs jump the same number of discrete speeds between depletion intervals (cf.~Section~\ref{sec:approach_overview}), that speed levels are non-decreasing, and that the currently processed job is one of maximum speed level among active jobs.
Using this definition, we are ready to characterize optimal schedules in terms of the following combinatorial properties.
Note that for~\ref{thm:optimality:b} of the following theorem, one can simply use a YDS schedule (cf.~\cite{Yao:1995}) for the workload assigned to the corresponding depletion interval.
\begin{theorem}\label{thm:optimality}
Consider a schedule $S$ and a recharge rate $R$.
The following properties are sufficient\footnote{%
	If we restrict ourselves to \emph{normalized} (earliest deadline first, only one speed change per job in a depletion interval) schedules, they are in fact also necessary.
} for $S$ and $R$ to be optimal:
\begin{enumerate}
\item\label{thm:optimality:a} $S$ is feasible.
\item\label{thm:optimality:b} The work in each depletion interval is scheduled energy optimal.
\item\label{thm:optimality:c} The SLR holds.
\item\label{thm:optimality:d} There is a split depletion point: a depletion point $\tau_k>0$ such that no job with deadline greater than $\tau_k$ is processed before $\tau_k$.
\end{enumerate}
\end{theorem}
We defer the proof to Appendix~\ref{sec:structural_optimality_conditions}.


\section{Notation}\label{sec:notation}
Given a schedule $S$, we need a few additional notions to describe and analyze our algorithm.
We defer any notation needed exclusively for proofs to the appendix.

\paragraph{Structuring the Input}
Let us start by formally defining depletion points and depletion intervals.
As noted earlier, depletion points represent time points where our algorithm maintains a battery level of zero and partition the time horizon into depletion intervals.
Note that these definitions depend on the current state of the algorithm.
\begin{definition}[Depletion Point]
Let $E_S(t)$ be the energy remaining at time $t$ in schedule $S$.
Then $\tau_i$ is a \emph{depletion point} if $E_S(\tau_i)=0$ (and the algorithm has labeled it as such).
$L$ is the number of depletion points, $\tau_0\coloneqq0$, and $\tau_{L+1}\coloneqq\infty$.
\end{definition}
\begin{definition}[Depletion Interval]
For $\ell>0$, the $\ell$-th depletion interval is $I_{\ell}\coloneqq\intco{\tau_{\ell-1},\tau_{\ell}}$.
We also define $s_{j,\ell}$ as the (average) speed of job $j$ during $I_{\ell}$.
\end{definition}
To simplify the discussion, we sometimes identify a depletion interval $I_{\ell}$ with its index $\ell$.
While moving work between depletion intervals, our algorithm uses the jobs' \emph{speed levels} together with the SLR as a guide:
\begin{definition}[Speed Level]
For all $j, \ell$ with $I_{\ell} \cap \intco{r_j,d_j} \neq \emptyset$, the \emph{speed level $\SL(j,\ell)$} of $j$ in $I_{\ell}$ is such that if $j$ is processed in $I_{\ell}$, then $s_{j,\ell} \in \intcc{s_{\SL(j,\ell)-1}, s_{\SL(j,\ell)}}$.
\end{definition}
Note that this definition should be understood as a variable of our algorithm.
In particular, it is not unique if the job runs at a discrete speed $s_{i-1}$.
In these cases, $\SL(j,\ell)$ can be either $i-1$ or $i$ (and the algorithm can set $\SL(j,\ell)$ as it wishes).
The algorithm initializes the speed level for every depletion interval where $j$ is active based on the initial YDS schedule and assigns speed levels maintaining the SLR throughout its execution.

Next, we give a slightly weaker version of the well-known EDF (Earliest Deadline First) scheduling policy (see Appendix~\ref{app:additional_notation} for the full definition).
The idea is to maintain EDF w.r.t.~depletion intervals but to allow deviations within depletion intervals.
For example, we avoid schedules with depletion intervals $I_1,\dots,I_4$ where job $j_1$ is scheduled in $I_1$ and $I_3$ and $j_2$ in $I_2$ and $I_4$.
This will ensure that the collection of transfer paths will be laminar, which is useful throughout the analysis.
\begin{definition}[Weak EDF, informal]
Schedule $S$ is \emph{weak EDF} if there is a schedule that is EDF in which each job is run in the same depletion intervals as in $S$.
\end{definition}
Next, we consider to what extent a schedule adheres to the optimality conditions (Theorem~\ref{thm:optimality}).
We distinguish between schedules that (essentially) adhere to the first two optimality conditions and schedules that also have the third optimality condition (SLR).
\begin{definition}[Nice \& Perfect]\label{def:nice_perfect}
Schedule $S$ is \emph{nice} if it is feasible, obeys YDS between depletion points, and satisfies weak EDF\@.
If, additionally, $S$ fulfills the SLR, we call it \emph{perfect}.
\end{definition}

\paragraph{Distributing Workload}
We now define \es, the building block for our algorithm.
Intuitively, they formalize possible ways to move work around between depletion intervals.
Our definition ensures that moving work over \es maintains niceness throughout the algorithm's execution.
Moreover, we also ensure that \es only affect the schedule's speed profile at their sources/targets.
\begin{definition}[\e]\label{def:eps}
The sequence $(\ell_a,j_a)_{a = 0}^s$ is called an \emph{\e} if we can, simultaneously for all $a$, move some non-zero workload of $j_a$ from $\ell_{a-1}$ to $\ell_a$ while maintaining niceness and without changing any job speeds in $\ell_1,\dots,\ell_{s-1}$.
The pair ($\ell_0,j_0$) (resp., $(\ell_s,j_s)$) is the \emph{source} and \emph{source job} (resp., \emph{destination} and \emph{destination job}) of the \e.
The \e is \emph{active} if it also maintains perfectness.
\end{definition}
Each edge drawn in Figure~\ref{fig:moving_work} is a (trivial) \e.
See Figure~\ref{fig:more_complex_epsilon_transfer} in Appendix~\ref{app:additional_example} for a more complex example of an \e.

Next we define the priority of an \e.
Our algorithm compares \es based on source and destination.
Once the source and destination have been fixed, the priority is used to determine which \e is used to transfer work.
As mentioned in Section~\ref{sec:approach_overview}, the basic idea is to: prefer transfers that move work right to transfers that move work left, between transfers moving work right prefer shorter transfers, and between transfers moving work left prefer longer transfers.
\begin{definition}[Transfer Priority]\label{def:eps:domination}
Let $T_1=(\ell^1_a,j^1_a)_{a=0}^{s_1}$ and $T_2=(\ell^2_a,j^2_a)_{a=0}^{s_2}$ be two different \es with $\ell^1_0=\ell^2_0$ and $\ell^1_{s_1}=\ell^2_{s_2}$.
Let $a_1^*=\arg\min_a\set{\ell^1_a\neq\ell^2_a}$ and $a_2^*=\arg\min_a\set{j^1_a\neq j^2_a}$.
We say that $T_1$ is \emph{higher priority} than $T_2$ if
\begin{enumerate}
\item $\ell^2_{a_1^*}<\ell^1_{a_1^*-1}<\ell^1_{a_1^*}$, or
\item $\ell^1_{a_1^*}<\ell^2_{a_1^*}$, and either $\ell^2_{a_1^*}<\ell^1_{a_1^*-1}$ or $\ell^1_{a_1^*-1}<\ell^1_{a_1^*}$, or
\item $a_1^*$ does not exist and the deadline for $j^1_{a_2^*}$ is earlier than the deadline for $j^2_{a_2^*}$.
\end{enumerate}
\end{definition}
Finally, we can define our multigraph of legal \es.
\begin{definition}[Distribution Graph]\label{def:distribution_graph}
The \emph{distribution graph} is a multigraph $G_D=(V_D,E_D)$.
$V_D$ is the set of depletion points and for every active \e $(\ell_a,j_a)_{a=0}^s$, there is a corresponding edge with source $\ell_0$ and destination $\ell_s$.
\end{definition}


\section{Algorithm Description}\label{sec:algorithm_description}
This section provides a formal description of the algorithm.
From a high level, the algorithm can be broken into two pieces:
\begin{enumerate*}
\item choosing which \es to move work along (in order to lower the recharge rate), and
\item handling events that cause any structural changes.
\end{enumerate*}
We start in Section~\ref{sec:handling_events} by describing the structural changes our algorithm keeps track of and by giving a short explanation of each event.
Section~\ref{subsec:algorithm_description} describes our algorithm.
Due to space constraints, the full correctness and runtime proofs are left for Appendix~\ref{sec:runtime_analysis} and~\ref{app:alg_correctness}.

\subsection{Keeping Track of Structural Changes}\label{sec:handling_events}
How much work is moved along each single \e depends inherently on the structure of the current schedule.
Thus, intuitively, an event is any structural change to the distribution graph or the corresponding schedule while we are moving work.
At any such event, our algorithm has to update the current schedule and distribution graph.
The following are the basic structural changes our algorithm keeps track of:
\begin{itemize}
\item \textbf{Depletion Point Appearance:}
	For some job $j$, the remaining energy $E_S(d_j)$ at $j$'s deadline becomes zero and the rate of change of energy at $d_j$ is strictly negative.
	If we were not to add this depletion point, the amount of energy available at $d_j$ would become negative, violating the schedule's feasibility.
	We can easily calculate when this happens by examining the rate of change of $R$ as well as the rate of change of $s_{j,\ell}$ for all jobs $j$ that run in the depletion interval $I_{\ell}$ containing $d_j$.

\item \textbf{Edge Removal:}
	An edge removal occurs when, for some job $j$, the workload of $j$ processed in a depletion interval $I_{\ell}$ becomes zero.
	In other words, all of $j$'s work has been moved out of $I_{\ell}$.
	Similar to before, we can easily keep track of the time when this occurs for any job $j$ processed in a given depletion interval, since all involved quantities change linearly.

\item \textbf{Edge Inactive:}
	An edge inactive event occurs when for some job $j$ its speed $s_{j,\ell}$ in a depletion interval $I_{\ell}$ becomes equal to some discrete speed $s_i$.
	Once more, we keep track of when this happens for each job processed in a given depletion interval.
\end{itemize}

\paragraph{Handling of Critical Intervals}
Note that by moving work along \es between two events $e_1$ and $e_2$, our algorithm causes
\begin{enumerate*}
\item the speed of exactly one YDS critical interval in each depletion interval to decrease and
\item the speed of some YDS critical intervals to increase.
\end{enumerate*}
For a single critical interval, these speed changes are monotone over time (between two events).
However, critical intervals might merge or separate during this process (e.g., when the speed of a decreasing interval becomes equal to a neighbouring interval).
In other words, the critical intervals of a given depletion interval might be different at events $e_1$ and $e_2$.
On first glance, this might seem problematic, as a critical interval merge/separation could cause a change in the rate of change of the critical interval's speed, perhaps with the result that the algorithm stops for events spuriously, or misses events it should have stopped for.
However, since only neighbouring critical intervals can merge and separate, this can be easily handled:
In any depletion interval, there are at most $\LDAUOmicron{n}$ critical intervals at event $e_1$.
Since only neighbouring critical intervals can merge/separate when going from $e_1$ to $e_2$, for each critical interval changing speed there are at most $\LDAUOmicron{n^2}$ possible candidate critical intervals that can be part of event $e_2$.
We just compute the next event caused by each of these candidates, and whether or not each candidate event can feasibly occur.
Then, the next event to be handled by our algorithm is simply the minimum of all feasible candidates.
This is an inefficient way to handle critical interval changes, but it significantly simplifies the algorithm description.
We leave the description of a more efficient way to handle \enquote{critical interval events} for the full version.

\paragraph{Handling Events}
When we have identified the next event, we must update the distribution graph and recalculate the rates at which we move work along the \es.
Given the definition of the distribution graph, updating the graph is fairly straightforward.
However, after updating the graph there might no longer be a path from every depletion interval to the far right depletion interval.
This can be seen as a cut in the distribution graph.
In these cases, to make progress, we either have to remove a depletion point or adapt the jobs' speed levels; If both of these fixes are not possible, our algorithm has found an optimal solution.
A detailed description of this can be found in Section~\ref{sec:handling_cuts} of the appendix.

\subsection{Main Algorithm}\label{subsec:algorithm_description}
Now that we have a description of each event type, we can formalize the main algorithm.
A formal description of the algorithm can be found in Listing~\ref{mainalgo}.
We give an informal description of its subroutines \textsc{CalculateRates}, \textsc{UpdateGraph}, and \textsc{PathFinding} below.

\begin{lstlisting}[float,keywords={for,while},caption={The algorithm for computing minimum recharge rate schedule.},label={mainalgo}]
Set $R$ to be recharge rate that ensures YDS schedule, $S$, is feasible
Let $G_D = (V_D,E_D)$ be the corresponding Distribution Graph.
$G'_D = \textsc{PathFinding}(G_D)$
($\Delta$, $\delta_{j,\ell}$, $T$)$\coloneqq\textsc{CalculateRates}(G'_D,S)$
while True:
	for each job $j$ and depletion interval $\ell$:
		set $s_{j,\ell} = s_{j,\ell} + \Delta \cdot \delta_{j,\ell}$
	set $R = R - \Delta$
	$\textsc{UpdateGraph}(T,G_D,S)$
	if $\exists$ fixable cut:
		fix cut with either a depletion point removal or SLR procedure
	else:   exit
	$G'_D = \textsc{PathFinding}(G_D)$
	($\Delta$, $\delta_{j,\ell}$, $T$)$\coloneqq\textsc{CalculateRates}(G'_D,S)$
\end{lstlisting}

\begin{description}
\item[\normalfont\textsc{UpdateGraph}$(T,G_D,S)$:] This subroutine takes an event type $T$, the distribution graph $G_D$ and the current schedule $S$ and performs the required structural changes.
It suffices to describe how to build the graph from scratch given a schedule (computing a schedule simply involves computing a YDS schedule between each depletion point).
Now the question becomes: Given two depletion points, how do we choose the \e between these two?
While perhaps daunting at first, this can be achieved via a depth-first search from the source depletion interval.
Whenever the algorithm runs into a depletion interval it has previously visited in the search, it chooses the higher priority \e of the two as defined by the priority relation.

\item[\normalfont\textsc{PathFinding}$(G_D)$:] We define \textsc{PathFinding}$(G_D)$ in Listing~\ref{pathfinding}.
Note the details of determining the highest priority edge are omitted but the implementation is rather straightforward.
The priority relation for choosing edges is: First choose the shortest right going edge, and otherwise choose the longest left going edge.
While this priority relation itself is rather straightforward, it requires a non-trivial amount of work to show that it yields suitable monotonicity properties to bound the runtime (see Appendix~\ref{sec:runtime_analysis}).

\begin{lstlisting}[float,keywords={for,while},caption={The \textsc{PathFinding} subroutine.},label={pathfinding}]
Let $S =\set{v_L}$, where $v_L$ is rightmost vertex
while exists an edge $e = (v_1,v_2)$ with $v_1,v_2\in S$ and $e$ is the highest priority such edge:
	add $v_1$ to S
\end{lstlisting}

\item[\normalfont\textsc{CalculateRates$(G'_D,S)$}:] This subroutine takes as input the set of paths from the distribution graph $G'_D$ and the current schedule $S$.
It returns for each job $j$ and each depletion interval $\ell$, the rate $\delta_{j,\ell}$ at which $s_{j,\ell}$ should change, $T$, the next event type, and $\Delta$ the amount the recharge rate should be decreased.
It is straightforward to see the set of paths chosen by the algorithm $G'_D$ can be viewed as a tree with the root being the rightmost depletion interval.
Assuming $R$ is decreasing at a rate of $1$, and working our way from the leaves to the root, we can calculate $\delta_{j,\ell}$ such that the rate of change of energy at all depletion points remains $0$.
With these rates, we can use the previously discussed methods to find both $T$ and $\Delta$.
\end{description}

\bibliographystyle{my-plainnat}
\bibliography{references}

\newpage
\appendix

\section{Additional Example}\label{app:additional_example}
\begin{figure}
\centering
\includegraphics[width=0.618\linewidth]{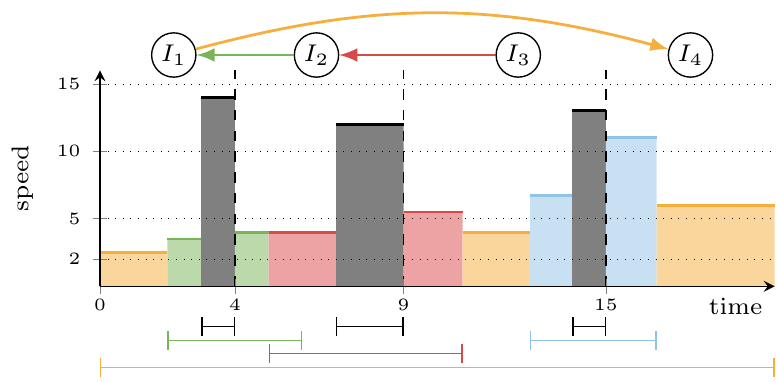}
\caption{%
	Variant of Figure~\ref{fig:moving_work}.
}
\label{fig:more_complex_epsilon_transfer}
\end{figure}
Figure~\ref{fig:more_complex_epsilon_transfer} is a variant of Figure~\ref{fig:moving_work} with the release-deadline interval of the green and red job changed such that they can form one combined \e.
Note that we can move work of the red job from $I_3$ to $I_2$ and a suitable amount of the green job from $I_2$ to $I_1$ such that their speeds in $I_2$ do not change (due to this combined movement) but the slice of the green job gets thinner while the slice of the red job gets thicker.
Also note that each edge is an \e of its own.
One (not necessarily optimal) way to decrease the recharge rate in this example is to move (net) work from $I_3$ to $I_1$ via the combined red-green \e, work from $I_2$ to $I_1$ via the green \e, and work from $I_1$ to $I_4$ via the orange \e.
Note that the last \e has to move enough work to make up for the additional work coming in from $I_2$ and $I_3$.

\section{Additional Notation}\label{app:additional_notation}
The following notation is used throughout proofs in the appendix but is not strictly needed for the rest of the paper.

\begin{definition}[Power Function Slopes]
For the speeds $s_i$ and their powers $P_i$ we define $\Delta_i\coloneqq\frac{P_i-P_{i-1}}{s_i-s_{i-1}}$.
\end{definition}
Remember that $\Delta_i=c^{i-1}\Delta_1$ for a constant $c>1$ by well-separation (cf.~Section~\ref{intro}).

The following two definitions are relatively technical, but essentially describe a weaker version of the EDF (Earliest Deadline First) scheduling policy.
\begin{definition} [first-run sequence]
Assume $d_1 < d_2 < d_3 < \dots < d_n$.
To construct the \emph{first-run}  sequence $(I_{\ell}^j)_{j = 1}^{k}$ of a depletion interval $I_{\ell}$ let $j_{\ell_1}, \dots , j_{\ell_k}$ be the $k$ jobs run in $I_{\ell}$ ordered by the first time they are run within $I_{\ell}$.
Then, $(I_{\ell}^j) = (\ell_1, \dots , \ell_k)$.
The first-run sequence of a schedule $S$ is the concatenation of all depletion interval first-run sequences from first to last.
\end{definition}

\begin{definition}[Weak EDF]
We say that a schedule satisfies \emph{weak EDF} if the corresponding first-run sequence $(S_k)$ has the following property.
For every $j \in [n]$, let $f_j$ and $l_j$ be the first and last appearances of $j$ in $(S_k)$.
Then, for all $i$ such that $f_j < i < l_j$, $S_i > j$.
\end{definition}

Finally, the last additional notion we'll be using captures when a job can move work to a given depletion interval.
This will be of particular importance when adjusting speed levels, as we have to make sure that these remain consistent.
\begin{definition}[Reachable]
A depletion interval $\ell$ is \emph{reachable} (resp., \emph{actively reachable}) by $j$ if there is an \e (resp., active \e) with source job $j$ and destination $\ell$.
\end{definition}

\section{Proof of Structural Optimality Conditions}\label{sec:structural_optimality_conditions}
In the following we provide the complementary slackness conditions obtained from the primal-dual formulation of our problem in Section~\ref{sec:structural_optimality_via_lp}.
Using these, we then prove Theorem~\ref{thm:optimality}.
\begin{align}
x_{jit}   &> 0   &   \Rightarrow   &&   \alpha_js_i-\sum_{t'\geq t}\beta_{t'}P_i-\gamma_t   &= 0    &,\label{eqn:complsl:1}\\
R         &> 0   &   \Rightarrow   &&   \sum_t\beta_tt                                      &= 1    &,\label{eqn:complsl:2}\\
\alpha_j  &> 0   &   \Rightarrow   &&   \sum_{t\in\intco{r_j,d_j}}\sum_ix_{jit}             &= p_j  &,\label{eqn:complsl:3}\\
\beta_t   &> 0   &   \Rightarrow   &&   \sum_{t'\leq t}\sum_{j\in J}\sum_{i=1}^kx_{jit'}P_i &= Rt   &,\label{eqn:complsl:4}\\
\gamma_t  &> 0   &   \Rightarrow   &&   \sum_{j\in J}\sum_{i=1}^kx_{jit}                    &= 1    &.\label{eqn:complsl:5}
\end{align}

\begin{proof}[of Theorem~\ref{thm:optimality}]
The feasibility of $S$ immediately gives us a set of candidate primal variables that fulfill Equation~\eqref{eqn:complsl:3} of the complementary slackness conditions.
Assuming time slots to be small enough also ensures that the $x$-variables are integral and that Equation~\eqref{eqn:complsl:5} is fulfilled\footnote{To see this, first note that we can normalize any (also an optimal) schedule using the EDF scheduling policy.
By choosing the time slots small enough, each job is processed alone and at a constant speed within a slot.
Note that we don't need to know the time slots' size for this argument, the mere existence of such time slots is sufficient (since the resulting optimality conditions are oblivious of the time slots).}.
In the following, we show how to define a set of feasible dual variables such that the remaining complementary slackness conditions hold.
This immediately implies optimality.

Before we define the dual variables, let us define some helper variables that describe how speed levels change at depletion points.
Fix a set of speed levels that adheres to the SLR\footnote{If we can choose, we choose the smallest possible speed level.} and consider the depletion points $0<\tau_1<\tau_2<\dots<\tau_L$ of $S$.
By~\ref{thm:optimality:d}, there is at least one depletion point $\tau_k$ such that no job $j$ with $d_j>\tau_k$ is processed in $\intco{0,\tau_k}$.
Without loss of generality, let $\tau_k$ be the leftmost depletion point with this property.
By this choice, for any depletion point $\tau_{\ell}$ with $\ell\in\set{1,2,\dots,k-1}$ there is a job $j_{\ell}$ that is active immediately before and after $\tau_{\ell}$.
The speed level of $j_{\ell}$ increases by $\SL(j_{\ell},\ell+1)-\SL(j_{\ell},l)\in\N_0$ from the $\ell$-th to the $\ell+1$-th depletion interval.
By the SLR, this increase is independent of the concrete choice of $j_{\ell}$ (any active job's speed level increases by the same amount from $I_{\ell}$ to $I_{\ell+1}$).
Thus, we can define $a_{\ell}\coloneqq\SL(j_{\ell},\ell+1)-\SL(j_{\ell},\ell)\in\N_0$ as the increase of speed level at $\tau_{\ell}$.

We are now ready to define the dual variables.
For $t\not\in\set{\tau_1,\tau_2,\dots,\tau_{k}}$ we set $\beta_t=0$.
The remaining $\beta_t$ variables are defined by the (unique) solution to the following system of linear equations:
\begin{equation}
\begin{aligned}
&\sum_{\ell'=\ell}^{k}\beta_{\tau_{\ell'}}=c^{a_{\ell}}\sum_{\ell'=\ell+1}^{k}\beta_{\tau_{\ell}}   &   \text{for }\ell\in\set{1,2,\dots,k-1},\\
&\sum_{\ell=1}^k\tau_{\ell}\beta_{t_{\ell}}=1
.
\end{aligned}
\end{equation}
Here, $c>1$ is the constant from the definition of well-separation (see Section~\ref{intro}).
That is, for all $i$ we have $\Delta_{i+1}(\coloneqq\frac{P_{i+1}-P_{i}}{s_{i+1}-s_{i}})=c\cdot\Delta_{i}$.
By construction, these $\beta_t$ variables fulfill Equations~\eqref{eqn:complsl:2} and~\eqref{eqn:complsl:4}.

It remains to define suitable $\alpha_j$ variables such that Equation~\eqref{eqn:complsl:1} holds and to show that these dual variables are feasible.
So fix a job $j$ and let $I_{\ell}$ be the first depletion interval in which $j$ is processed.
We set $\alpha_j=\Delta_{\SL(j,\ell)}\cdot\sum_{t\geq\tau_{\ell}}\beta_t$.
Using a simple induction together with the definition of the $\beta_{t}$ and $a_{\ell}$, this implies $\alpha_{j}=\Delta_{\SL(j,\ell')}\cdot\sum_{t\geq\tau_{\ell'}}\beta_{t}$ for all $\ell'\geq\ell$.
To set $\gamma_t$, remember that we assume time slots to be small enough that at most one job is processed.
If no job is processed, we set $\gamma_t=0$.
Otherwise, let $j$ be the job processed in $t$ and $\ell$ the depletion interval that contains $t$.
We set $\gamma_t=\alpha_j s_{\SL(j,\ell)}-\sum_{t'\geq t}\beta_{t'}P_{\SL(j,\ell)}$.
By construction, these variables fulfill the remaining complementary slackness condition (Equation~\eqref{eqn:complsl:1}), $\alpha_{j}\geq0$, $\beta_{t}\geq0$, and the second dual constraint holds.
Thus, it remains to show that for all $t$ we have $\gamma_t\geq0$ and that for all $j$, $i$, and $t$ we have $\alpha_js_i-\sum_{t'\geq t}\beta_{t'}P_i-\gamma_t\leq0$.
For the inequality $\gamma_t\geq0$, let $j$ be the job processed in $t$ (if there is no job, we have $\gamma_t=0$ by definition) and let $I_l$ be the depletion interval that contains $t$.
Note that $\sum_{t'\geq t}\beta_{t'}=\sum_{t'\geq\tau_{\ell}}\beta_{t'}$.
Then the desired inequality follows from
\begin{align}
     \frac{\gamma_t}{s_{\SL(j,\ell)}\sum_{t'\geq t}\beta_{t'}}
=    \frac{\alpha_j}{\sum_{t'\geq \tau_{\ell}}\beta_{t'}}-\frac{P_{\SL(j,\ell)}}{s_{\SL(j,\ell)}}
=    \Delta_{\SL(j,\ell)}-\frac{P_{\SL(j,\ell)}}{s_{\SL(j,\ell)}}
\geq 0
\end{align}
(the last inequality follows since speeds are well-separated).
Now fix a job $j$, a speed index $i$, and a time slot $t$ in which $j$ is active.
Let $l$ denote the depletion interval that includes $t$ and let $j'$ be the job that is actually processed in $t$.
We have to show $\alpha_js_i-\sum_{t'\geq t}\beta_{t'}P_i\leq\gamma_t$.
This is trivial if $\alpha_j=0$.
Otherwise, using the definition of $\alpha_j$ and $\gamma_t$ and dividing by $\sum_{t'\geq t}\beta_{t'}$, it is equivalent to
\begin{align*}
&\iff   &   \Delta_{\SL(j,\ell)}s_i-P_i &\leq \Delta_{\SL(j',\ell)}s_{\SL(j',\ell)}-P_{\SL(j',\ell)}\\
&\iff   &   P_{\SL(j',\ell)}-P_i        &\leq \Delta_{\SL(j',\ell)}s_{\SL(j',\ell)}-\Delta_{\SL(j,\ell)}s_i
.
\end{align*}
Since we assume YDS is used in between depletion points, we know $j'$ runs at least as fast in $I_{\ell}$ as $j$.
Thus we have $s_{j',\ell}\geq s_{j,\ell}$ and, in particular, $\SL(j',\ell)\geq\SL(j,\ell)$.
Thus, it is sufficient to show $P_{\SL(j',\ell)}-P_i\leq\Delta_{\SL(j',\ell)}(s_{\SL(j',\ell)}-s_i)$, which follows once more from the well-separation of the speeds.
\qed
\end{proof}


\section{Runtime Analysis}\label{sec:runtime_analysis}
In this section we provide an analysis on the runtime and correctness of our algorithm.
We begin with some notes on how the algorithm handles certain cases,
We then bound the number of different events that can occur.
Finally, we analyze the runtime of the calculations made by our algorithm in between events.

\subsection{Intricacies of the Algorithm}

In this subsection, we describe informally some intricacies of the algorithm that, though not vital to a high-level understanding of the algorithm, are key in its formal analysis.

\paragraph{Valid \es}
The definition of \e allows for many counterintuitive \es:
For example, ones that take the same job multiple times, or enter the same critical interval multiple times.
It is easy for the depth-first search that chooses \es to prune such undesirable \es.
Here we describe the set of \es pruned, and why.

\begin{itemize}
\item \textbf{\es that take the same job multiple times, or enter the same critical interval multiple times.}  It is easy to see that such \es are in some sense not minimal, and an \e with strictly fewer edges could be obtained.
\item \textbf{\es violating weak EDF.}
This allows us to say the subgraph of the distribution graph taken by the algorithm has edges that are laminar, and that the algorithm never has reason to take \es that cross each other, as well as that the edges of \es taken by our algorithm are laminar.
\item \textbf{\es that take a right edge that is completely contained within a previously taken left edge.} Though less intuitive, one can show that such \es can be replaced by a series of left \es, in a manner similar to that used in the proof of Lemma~\ref{lem:composition} below.
By disallowing such \es, we can say more about the \es taken by the algorithm, making analysis simpler.
\item \textbf{\es that take work of a job in the opposite direction of previously taken \es.}
This allows us to formally prove that any job's workload moves in two phases in between two cut events: first only right and then only left.
This insight is key to bounding the number of edge removal events.
\end{itemize}

\paragraph{Avoiding Critical Interval Events}
As described, the algorithm does not stop for critical interval events.
To gain some insight into how this is accomplished, we briefly describe how to calculate the speed level event where a critical interval's speed becomes the upper speed of its current speed level.
In a depletion interval, there may be multiple jobs at a speed level, with different releases and deadlines, and so multiple possible ways such a speed level event can occur.
However, we know that if this event occurs, it occurs at a critical interval whose borders are releases (or the first time the job can be run in the depletion interval, according to SLR) or deadlines.
Thus, for every pair of releases and deadlines, we can compute which jobs must be run in that interval within the depletion interval.
If it were the case that the next event is in fact a speed level event caused by this critical interval, we can calculate when it would occur by looking at the rate at which just these jobs are getting work, and calculating when the resulting critical interval speed would become the maximum for its speed level.
By considering all possible critical intervals, we can determine which speed level event will actually occur first.
For other types of events, similar calculations can be performed.

\subsection{Bounding the Runtime}
Here we show that the number of events that cause the algorithm to recalculate are bounded by a polynomial.
The idea is to first bound the number of \emph{cut events}: situations, in which there is a depletion interval without a path of \es to move work to the far right depletion interval (see Section~\ref{sec:handling_cuts} for details).
Then, we show that between any two cut events, there are only a polynomial number of other events.
While using such a hierarchical structure to bound events may artificially increase the runtime bound, it is helpful in simplifying the analysis.
We conclude this subsection with Theorem~\ref{thm:algo_runtime}, which bounds the total runtime.

\begin{lemma}\label{lem:runtime_cut_fixing}
The algorithm fixes a cut in the distribution graph at most $k n^2$ times.
\end{lemma}

\begin{proof}
Observe that Property 1 (c) of Theorem~\ref{thm:optimality} can be restated as follows: for each non-degenerate depletion point $i$ (ordered from left to right) there exists a number $\delta_i \in [k] \cup \{0\}$ such that for any job $j$ and any depletion intervals $\ell_1$ and $\ell_2$, $\ell_1 < \ell_2$, in which $j$ is active, we have that
\[
\mathcal{L}(j,\ell_2) -\mathcal{L}(j,\ell_2) = \sum_{i=\ell_1}^{\ell_2-1} \delta_i.
\]
The algorithm assigns speed levels to jobs in intervals in which they are not active such that this definition is satisfied whenever $j$ is alive.

For any intermediate schedule $\mathcal{S}$ produced by the algorithm, order jobs by increasing deadline, and let $\delta_{j,\mathcal{S}} = \delta_i$ for the job $j$ with smallest index whose deadline is the same as the $i$th depletion point, and $\delta_{j,\mathcal{S}} = 0$ otherwise.
Consider the following potential function:
\[
\Phi(\mathcal{S}) = \sum_{j=1}^n j \cdot \delta_{j,\mathcal{S}}
\]
Recall that speed levels are only modified when a cut in the Transfer Graph is fixed, so this is the only time that the $\delta_i$ change.
For a fixed cut, let $l$ be the index of the left depletion point defining the cut (if it exists), and $r$ be the index of the right depletion point defining the cut (which always exists).
It is straightforward to observe that the algorithm's modification of speed levels to fix a cut increases $\delta_r$ by $1$ and, if $l$ exists, decreases $\delta_l$ by $1$.
Thus, since there is at most one depletion point per time, $\Phi$ increases by at least one every time a cut in the Transfer Graph is fixed (and this is the only event that changes $\Phi$).
It is also clear that $0 \leq \Phi \leq k n^2$, since for any $i$, $\delta_i \leq k$, and the Lemma follows.
\qed
\end{proof}

\paragraph{Bounding Events via Isolation}
Before we can bound the number of events that occur between cut events, we need several auxiliary results.
These form the most technical result of the paper, but turn out to provide strong tools, such that bounding the actual events later on will be relatively straightforward.
We first provide some results about our choice of \es.
The most important part will be when we introduce \emph{isolated areas}.
Intuitively, we will show that during the executing of our algorithm, some time intervals will become isolated in the sense that no workload enters them and any workload that leaves them can do so only in a very restricted way.
This turns out a strong monotonic property that helps to bound the number of events.

We say the source (resp., destination) is \emph{outside} an interval $[t_1,t_2]$ if the critical interval the source job (resp., destination job) is running in does not intersect $[t_1,t_2]$, and it is \emph{inside} otherwise.
\begin{observation}
\label{obs:pathOfTransfers}
When the algorithm chooses an \e $T$ with source $s_T$ and destination $d_T$, there is a path of \es from the source of that \e to $I_L$.
If $T$ is a left \e, no previously chosen \e has source or destination within $(d_T, s_T]$.
If $T$ is a right \e, no previously chosen \e with source or destination within $[s_T,d_T)$ is a left \e.
\end{observation}

\begin{lemma}
\label{lem:transferOrder}
If two critical intervals $C_1$ and $C_2$ in the same depletion interval, with $C_1$ to the left of $C_2$, both have active \es to the same destination depletion interval $\ell$, then there is an active \e from $C_2$ with destination $\ell$ that is higher priority than all active \es from $C_1$ with destination $\ell$.
\end{lemma}
\begin{proof}
This follows easily from the priority definition of \es in Definition~\ref{def:eps}.
\qed
\end{proof}

\begin{definition}[\e span and crossing]
For an \e $T$, let $l$ be the time that the leftmost of source and destination critical intervals of $T$ begins, and $r$ be the time that the rightmost of source and destination critical intervals of $T$ ends.
Then $[l,r]$ is the \emph{span} of $T$.
Two \es $T_1$ and $T_2$ are \emph{crossing} if their source and destination critical intervals are all unique, and the intersection of their spans is nonempty.
\end{definition}

\begin{lemma}
\label{lem:composition}
Let $T_1$ and $T_2$ be two \es that cross.
Then there is an path of \es $T_3$ with source that of $T_1$ and destination that of $T_2$, and every intermediate destination and source is active.
Additionally, if the source of $T_1$ is decreasable, and the destination of $T_2$ is increasable, then $T_3$ is active.
\end{lemma}
\begin{proof}
We split the proof into cases, based on the directions of $T_1$ and $T_2$, and their sources and destinations.
Let $a_f$ be the final index of depletion intervals of $T_2$.
We illustrate only one case, as the remaining cases use essentially the same arguments.

\textbf{\boldmath Case 1: $T_1$ is a right \e and $T_2$ is a right \e, and the source of $T_1$ is left of the source of $T_2$.}
Let $a_1$ be the lowest index such that $\ell_{a_1}$ of $T_1$ is or is to the right of the critical interval from $\ell_{a_2}$ of $T_2$, and to the left of the critical interval of $\ell_{a_2+1}$ of $T_2$, which must exist since $T_1$ and $T_2$ are crossing.
Let $e$ be the edge from $\ell_{a_2}$ to $\ell_{a_2+1}$ in $T_2$.
If the $T_1$ critical interval in $\ell_{a_1}$ is the $T_2$ critical interval in $\ell_{a_2}$, then it is clear that we can create $T_3 = (\ell_a,j_1)_{a=0}^{a_1-1} \cup  (\ell_a,j_a)_{a=a_2}^{a_f}$.
Otherwise, it must be that $\ell_{a_1} \neq \ell_{a_2}$, since they would have to run in the same critical interval otherwise, as the deadline of $j_{a_2}$ cannot be before $\ell_{a_2}+1$, and the release time of $j_{a_1}$ cannot be after $\ell_{a_2}$.
Note also that $j_{a_1}$ must have a later deadline than $j_{a_2}$, since it can move to $\ell_{a_1}$ without violating weak EDF, and similarly, $j_{a_2}$ completes in or before $\ell_{a_1}$.
Let $T' = (\tilde \ell_a,\tilde j_a)_{a=0}^k$ be the longest path of \es such that each $j_a$ has earlier deadline than $j_{a_1}$, and each $j_a$ can move work into the critical interval where $j_{a}-1$ completes.
Either such a $T'$ exists, or $j_{a_1}$ can be run in $\ell_{a_2+1}$ (in which case we can take $T'$ to be empty).
Thus we obtain $T_3 = (\ell_a,j_1)_{a=0}^{a_1-1} \cup (\ell_{a_1-1},j_{a_1-1}) \cup T'  (\ell_a,j_a)_{a=a_2}^{a_f}$ (which is possibly non-minimal, but can be reduced in size).
\qed
\end{proof}

\begin{lemma}
\label{lem:crossing}
If at any event, the algorithm chooses \es $T_1$ and $T_2$, then $T_1$ and $T_2$ are not crossing.
\end{lemma}
\begin{proof}
This follows as an easy consequence of both Lemma~\ref{lem:composition} as well as the definition of weak EDF.
\qed
\end{proof}

We now introduce notation and a definition that will be helpful in the coming proofs.
Fix any cut event, and let $\Gamma = \{1,\dots,\tau\}$ denote the events, in order, that the algorithm stops for between that event and the next cut event.
\begin{definition}[Isolated Area]
Let $\gamma \in \Gamma$ and $t$ be some time in the schedule, and $\ell_t$ be the depletion interval containing $t$.
Then the interval $[t,t']$ is an \emph{isolated area}, denoted by $\nu_t(\gamma)$
if it is possible to assign speeds to jobs in $\ell_t$ such that they obey their releases and deadlines, are run in earliest deadline first order, and $t'$ is a depletion point after $t$ such that there is no active \e with source inside $[t,t']$, and destination outside $[t,t']$, without crossing $t$ in $\ell_t$ (i.e., any \e with destination in $\ell_t$ can place work to the right of $t$ in $\ell_t$, and any \e with destination in a depletion interval either before $\ell_t$ or after $t'$ can be split into two active halves: one with destination in $\ell_t$ to the right of $t$, and the other with source $\ell_t$ to the left of $t$).
Additionally, the isolated interval is \emph{maximal} if $t'$ is the latest deplation point satisfying this definition.
$\ell_t$ is referred to as the \emph{exit} of $\nu_t(\gamma)$.
\end{definition}

\begin{observation}
\label{obs:isolatedExit}
For any isolated area, the \e $T$ chosen by the algorithm whose source is exit of the isolated area must have destination outside the isolated area.
Additionally, the path of \es to $I_L$ from any depletion interval in the isolated area must include $T$.
\end{observation}

\begin{lemma}
\label{lem:isolatedProperty}
For any maximal isolated area $\nu_t(\gamma)$, no \e taken by the algorithm has source outside $\nu_t(\gamma)$ and destination inside $\nu_t(\gamma)$.
\end{lemma}

\begin{proof}
Consider any active \e $T$ with source $s_T$ outside of $\nu_t(\gamma)$ and destination inside $\nu_t(\gamma)$.
We show that the algorithm does not choose $T$:
\begin{itemize}
\item \textbf{\boldmath Case 1: The source is leftmost depletion interval intersecting $\nu_t(\gamma)$.}
This follows immediately from Observation~\ref{obs:isolatedExit}.
\item \textbf{\boldmath Case 2: The source is to the left of the exit of $\nu_t(\gamma)$.}
Let $T_e$ be the \e with source the exit of $\nu_t(\gamma)$.
By Observations~\ref{obs:pathOfTransfers} and~\ref{obs:isolatedExit}, if the algorithm chose $T$, it must have chosen $T_e$ first.
If $T_e$ is a left \e, this contradicts Obersvation~\ref{obs:pathOfTransfers}.
If $T_e$ is a right \e, by Observation~\ref{obs:isolatedExit}, this contradicts Lemma~\ref{lem:crossing}.
\item \textbf{\boldmath Case 3: The source is to the right of $\nu_t(\gamma)$.}
We show that either the algorithm does not take $T$, or the right border of $\nu_t(\gamma)$ could be extended.
Let $I_r$ be the rightmost depletion interval that can be reached by a path of active right \es from $s_T$, and $I_l$ be the first depletion interval to the right of $\nu_t(\gamma)$.
\begin{enumerate}
\item If $I_r = I_L$, then the algorithm would take a right \e from $s_T$.
\item If there does not exist an active \e with source between $I_l$ and $I_r$, and destination to the left of the exit of $\nu_t(\gamma)$, then $\nu_t(\gamma)$ could extend to $I_r$, contradicting the definition of $\nu_t(\gamma)$.
\item If there exists an active \e with source between $s_T$ and $I_r$, and destination to the left of the exit of $\nu_t(\gamma)$, then the longest such \e would be taken before $T$, and the path of right \es from $s_T$ would be taken rather than $T$.
\item If there exists an active \e with source between $I_l$ and $s_T$, and destination to the left of the exit of $\nu_t(\gamma)$, then by Lemma~\ref{lem:composition}, there exists an active \e from $s_T$ to the left of the exit of $\nu_t(\gamma)$, and this \e is higher priority than $T$, so the algorithm does not take $T$.
\qed
\end{enumerate}
\end{itemize}
\end{proof}

\begin{lemma}
\label{lem:isolated}
Let $\nu_t(\gamma)$ be an isolated area.
If $\nu_t(\gamma)$ exists and is nonempty, then for any $\gamma' \in \Gamma$ with $\gamma' > \gamma$, the maximal isolated area $\nu_t(\gamma')$ exists and $\nu_t(\gamma') \supseteq \nu_t(\gamma)$.
\end{lemma}

\begin{proof}
We proceed by induction on events.
The base case, for event $\gamma$, follows by definition.
For the inductive step, suppose the lemma holds at event $\eta$, and we will show that the lemma continues to hold at event $\eta+1$.
We accomplish this by showing that any depletion interval that is part of $\nu_t(\eta)$ must be part of $\nu_t(\eta+1)$, and thus $\nu_t(\eta) \subseteq \nu_t(\eta+1)$.
We first consider the movement of work between events, and second consider the effect of the event $\eta+1$.

\textbf{Work Movement.}
Assume work is moved between $\eta$ and $\eta+1$.
We (conceptually) stop the algorithm just before enough work is moved to cause $\eta+1$.
We show that if there is an active \e with destination outside of $\nu_t(\eta)$ at this time (for ease, we write at $\eta+1$), then either the destination is part of $\nu_t(\eta+1)$, or some active \e with source inside $\nu_t(\eta)$ and destination outside $\nu_t(\eta)$ existed at $\eta$, contradicting that $\nu_t(\eta)$ is an isolated area.
Let $T = (\ell_a, j_a)_{a=0}^s$ be an active \e with source in $\nu_t(\eta)$ and destination outside $\nu_t(\eta)$ at $\eta+1$, and for $a=0,\dots,s$ let $C_a$ be the critical interval in $\ell_a$ used by that edge of the \e.

We show that an \e $R$ with source inside $\nu_t(\eta)$ and destination outside $\nu_t(\eta)$ must exist at $\eta$, and then show that an \e $R'$ with the same properties was active at $\eta$.
We consider transfer edges $i$ (taking work from $\ell_i$ to $\ell_i+1$) one at a time beginning with $s-1$ down to $0$.
Let $C_d$ be the current destination critical interval for the \e we are building, which is initially $C_s$.
We show that, assuming there is an \e from $C_{i+1}$ to $C_d$, then there is one from $C_{i}$ to either $C_d$ or some other critical interval outside $\nu_t(\eta)$ (i.e., we choose a new $C_d$).

First note that if $j_i$ was present in $C_i$ at $\eta$, edge $i$ must exist at $\eta$ as well (since no event happened), and thus an \e from $C_i$ to $C_d$ exists.
Otherwise, some \e $T'$ was taken at $\eta$ must have moved job $j$ into $C_i$.
If the source of $T'$ is in $\nu_t(\eta)$, then this source has an \e to $C_{i+1}$, and therefore an \e to $C_d$, thus we have found $R$.
If the source of $T'$ is outside $\nu_t(\eta)$, then the destination of $T'$ is also outside $\nu_t(\eta)$ by Lemma~\ref{lem:isolatedProperty}.
Thus there is an \e from $C_i$ to the destination of $T'$, so we change $C_d$ to be the destination critical interval of $T'$.

It remains to construct an \e $R'$ that is active at $\eta$.
Let $s_R$ and $d_R$ be the source and destination critical intervals of $R$.
We first construct an \e $R_1$ with destination $d_R$ with a decreasible source at $\eta$.
If $s_R$ is decreasible at $\eta$, then $R_1 = R$.
Otherwise, we know $s_R$ is decreasible at $\eta+1$, so if $s_R$ was not decreasible at $\eta$, the algorithm must have taken some \e $R_1'$ with destination $s_R$ at $\eta$, and the source of $R_1'$ must be in $\nu_t(\eta)$ by definition of isolated area.
Combine $R_1'$ and $R$ to yield $R_1$, an \e whose source is decreasible.
Similarly, if $d_R$ is decreasible at $\eta$, then $R' = R_1$.
Otherwise, we know $d_R$ is increasible at $\eta+1$, so if $d_R$ was not increasible at $\eta$, the algorithm must have taken some \e $R_2'$ with source $d_R$ at $\eta$, and the destination of $R_2'$ must be in ourside of $\nu_t(\eta)$ by Lemma~\ref{lem:isolatedProperty}.
Thus, combine $R_1$ and $R_2'$ to obtain $R'$

\textbf{Events.}
We show that the event $\eta+1$ cannot cause the isolated interval to decrease in size or cease to exist.
In each case, we show no new active \es with source inside $\nu_t(\eta)$ and destination outside $\nu_t(\eta)$ could become available.
\begin{itemize}
\item \textbf{Depletion Point Addition Events:}
For any depletion point added, the only \es affected are those with source or destination in the depletion interval that gained the depletion point, and whether or not those \es were active did not change.
\item \textbf{Depletion Point Removal Events:}
This even has a similar effect on \es as depletion point addition events, except when the depletion point removed is the rightmost depletion point of $\nu_t(\eta)$.
In this latter case, let $\ell_R$ be the depletion interval that merged with the rightmost depletion interval of $\nu_t(\eta)$.
The fact that the depletion point was removed means that there are now no active \es from $\ell_R$ to any other depletion interval, thus at $\eta+1$ the isolated area can be expanded to include the (now removed) depletion interval $\ell_R$.
\item \textbf{Edge Inactive Events:}
These cause \es to cease to be active, thus no new \es can appear as a result of these events.
\item \textbf{Critical Interval Merge and Separation Events:}
Critical intervals merging and separating can combine or split \es, but do not cause \es to become active from inactive.
Thus, the only place where merge events can cause a new active \e is at the exit of the isolated area; However, these new \es must cross $t$ and thus do not cause the isolated area to cease to exist.
\item \textbf{Edge Removal Events:}
These can only remove destinations for \es, and thus can only enlarge the isolated area.
\qed
\end{itemize}
\end{proof}

\paragraph{The Bounds on Events}
With these technical Lemmas we are now ready to bound the number of non cut events.
The hierarchy used assumes depletion point additions/removals and speed level events occur at the same level, but below cuts, and that edge removals occur at the bottom of the hierarchy.

\begin{lemma}
There are at most $O(n)$ depletion point addition and removal events.
\end{lemma}

\begin{proof}
We show that, between cut events, once a depletion point is removed, it never returns.
Since there are at most $n$ depletion points in the schedule, there can be at most $2n$ depletion point addition or removal events.
Intuitively, the removal of a depletion point creates an isolated area, which by Lemma~\ref{lem:isolated} persists.
We then argue that, since work is never removed from the right of the old depletion point, no new depletion point can appear there.

Suppose at event $\gamma$ a depletion point is removed at $t$, and let $t'$ be the time of the next depletion point.
We show that $[t,t']$ is an isolated area.
This follows because we do not remove a depletion point unless there is no active \e from the corresponding depletion interval to outside of it.

Fix any depletion interval, and observe that, as work is moved by the algorithm, the total energy available at any time point to the right of the critical interval being decreased must be increasing, due to the fact that the recharge rate is decreasing and the next depletion point must be maintained.
Since for any $\gamma'>\gamma$, by Lemma~\ref{lem:isolated} $\nu_t(\gamma')$ exists, and by the fact that the algorithm does not merge critical intervals that appear on both sides of $t$, no critical interval to the right of $t$ is ever the source of an \e, and thus the energy at $t$ is always increasing, and so $t$ can never be a depletion point again.
\qed
\end{proof}

\begin{definition}[work barrier]
\label{lem:workbarrier}
For a time $t$, a depletion point $t'$ is a \emph{$t$ work barrier} if there is no active \e with source in $[t,t']$ and destination to the right of $t'$.
\end{definition}

\begin{lemma}
\label{lem:workbarrierPersistence}
If $t'$ is a $t$ work barrier at $\gamma$ caused by a right \e as described in Lemma~\ref{lem:workbarrier}, then for any $\gamma' > \gamma$, there is some $t_2 \geq t'$ such that $t_2$ is a $t$ work barrier at $\gamma'$.
\end{lemma}

\begin{proof}
It is easy to see that for any isolated area $\nu_t(\gamma)$, the right endpoint is a $t$ work barrier.
We show that the isolated area $\nu_t(\gamma)$ exists, and thus the work barrier exists at $\gamma'$ as the right endpoint of $\nu_t(\gamma')$.

Assume the work barrier iss caused by a right \e $T$ taken by the algorithm, then suppose to obtain a contradiction that $[t,t']$ is not an isolated area.
Then there is some active \e $T'$ in $[t,t']$ with destination to the left of $t$, which is the right endpoint of the source of $T$.
By Lemma~\ref{lem:composition}, we can create an active \e from the source of $T'$ to the destination of $T$, contradicting there is a work barrier at $t'$.
\qed
\end{proof}

\begin{lemma}
\label{lem:speedlevelEvents}
There are at most $O(n)$ speed level events.
\end{lemma}

\begin{proof}
Let $\gamma$ be a lower speed level event, and $t$ be the left border of the critical interval causing this event.
We argue that there is an isolated area $\nu_t(\gamma)$.
As a result, no job that could be placed to the right of $t$ will ever be the source of an \e again, and thus this critical interval will never cause itself to decrease again (which could cause it to hit a lower speed level again, or cause it to not be at an upper speed level).
Since there are at most $O(n)$ critical intervals, there are at most $O(n)$ such events.

We now show that $\nu_t(\gamma)$ exists.
\begin{itemize}
\item \textbf{Case 1: Lower Speed Level Events.}
Consider the \e $T$, with source $C_s$ that was taken from the critical interval causing the lower speed level event at $\gamma-1$.
Let $t$ be the left endpoint of $C_s$.
There are two cases, depending on the direction of $T$.
\begin{itemize}
\item \textbf{\boldmath Subcase 1: $T$ is a right \e.}
By Lemma~\ref{lem:workbarrier}, if $t_1$ is the right endpoint of $C_s$, there was a $t_1$ work barrier at some depletion point $t'$, and thus there was no active \e from the right of $C_s$ to the right of $t'$.
Since $C_s$ hit a lower speed level, and no other event occurred, at $\gamma$ there is no active \e from the start of $C_s$ to the right of $t'$.
If we can show that there is no active \e from $[t,t']$ to the left of $t$, then we have shown that $\nu_t(\gamma)$ exists.
If such an active \e did exist, then it would have crossed $T$.
By Lemma~\ref{lem:composition} and Lemma~\ref{lem:transferOrder}, we would could construct a higher priority active \e to the destination of $T$ than $T$, contradicting that the algorithm took $T$.

\item \textbf{\boldmath Subcase 2: $T$ is a left \e.}
Let $t'$ be the work barrier caused by the parent of $T$ ($T$ must have a parent, since it is left-going and there is a path from the destination of $T$ to $\ell_L$).
By the definition of work barrier, there is no active \e in $[t,t']$ with destination to the right of $t'$.
If there were an active \e between $t$ and $t'$ with destination to the left of $t$, then by Lemma~\ref{lem:composition} we could construct $T'$ with source to the right of the source of $T$, and the same destination as $T$, and thus $T'$ would have been higher priority than $T$ by Lemma~\ref{lem:transferOrder}, so the algorithm would have taken it instead of, or as the parent of, $T$.
Thus $[t,t']$ is an isolated area.
\end{itemize}

\item \textbf{Case 2: Upper Speed Level Events.}
Consider the longest \e $T$, with source $C_s$ and destination $C_d$ where $C_d$ is the critical interval causing the upper speed level event at $\gamma-1$.
There are two cases, depending on the direction of $T$.
\begin{itemize}
\item \textbf{\boldmath Subcase 1: $T$ is a right \e.}
Consder the first event that causes $C_d$ to decrease again, and let $T'$ be the \e with source $C_d$.
At this event, if $T'$ is a right \e, let $t$ be the right endpoint of $C_d$.
Then there is a $t$ work barrier.
If $T'$ is a left \e, then if $t$ is the depletion point immediately to the right of $C_d$, there is a $t$ work barrier to the right of $C_s$.
In both cases, by Lemma~\ref{lem:workbarrierPersistence}, this work barrier persists, and thus any right \e with destination $C_d$ must be going to the sink of a left \e, contradicting that the algorithm would have chosen it.

\item \textbf{\boldmath Subcase 2: $T$ is a left \e.}
Let $t$ be the left endpoint of $C_d$.
Let $t'$ be the work barrier caused by the parent of $T$ ($T$ must have a parent, since it is left-going and there is a path from the destination of $T$ to $\ell_L$).
By the definition of work barrier, there is no active \e in $[t,t']$ with destination to the right of $t'$.
If there were an active \e between $t$ and $t'$ with destination to the left of $t$, it can be composed with $T$ by Lemma~\ref{lem:composition} to obtain a higher priority \e than $T$ that the algorithm could have chosen.
Thus $[t,t']$ is an isolated area.
\qed
\end{itemize}
\end{itemize}
\end{proof}

\begin{definition}[\boldmath$(j,\ell)$ (active) work barrier]
Let $t_1$ be the first time that $j$ can be run in $\ell$ (according to the speed levels of jobs in $\ell$).
A time (depletion point) $t'$ is a \emph{$(j,\ell)$ work barrier} (\emph{$(j,\ell)$ active work barrier}) if no job $j'$ with release time after $t_1$, and deadline before that of $j$, is part of an \e (path of active right \es) crossing $t'$ that does not contain an edge taking $j$, or some other job with earlier deadline than $j$ and release time before $t_1$, from $\ell$.
\end{definition}

\begin{lemma}
\label{lem:jworkbarrier}
Suppose the algorithm chooses an \e $T$ that contains an edge moving $j$ from $\ell_1$ to $\ell_2$.
Then there is a $(j,\ell_1)$ (active) work barrier somewhere to the right of $\ell_1$ and to the left of where $j$ is run in $\ell_2$ (to the right of $\ell_1$).
\end{lemma}

\begin{proof}
We prove the existence of the $(j,\ell_1)$ work barrier first.
For the sake of contradiction, suppose this does not hold, i.e., there exists \e $T'$ that takes some job $j'$ with release time after the first time $j$ can run in $\ell_1$, and with deadline before $j$, that does not contain an edge taking $j$, or some other job appropriate job, from $\ell_1$.
Then by Lemma~\ref{lem:composition}, we can compose the pieces of $T'$ and $T$ together to get an \e $R$ taking $j'$ (eventually) to $\ell_2$, to the destination of $T$.
We will show that $R$ is active and higher priority than $T$, contradicting that the algorithm chose $T$.
First note that the edge of $T'$ taking $j'$ must be right going, as the release of $j'$ is within $\ell_1$.
Additionally, $j'$ must be at the same speed level as $j$, or else it would not be able to cross the first time $j$ can run in $\ell_2$.

If $j$ and $j'$ are in the same critical interval, then we can create an \e from the source of $T$ to the destination of $T$ taking $j'$ instead of $j$ at $\ell_1$, which is clearly higher priority than $T$, as $j'$ is released later and has deadline earlier than $j$.
Otherwise, $j'$ is not at a lower speed level in $\ell_1$ since it must be at a higher speed than $j$, so we can take an \e with $j'$ as the source, which is higher priority than $T$ as long as the source of $T$ is $\ell_1$ or to the left of $\ell_1$.
If the source of $T$ is to the right of $\ell_1$, then it must be before $\ell_2$, since right edges cannot be below left edges of \es.
Note that the edge from $T$ into $\ell_1$ before taking $j$ must had source $\ell'$ to the right of the destination of the edge taking $j'$, or else there would be a weak EDF violation.
However, $T'$ must cross $\ell'$, since the destination of the $j$ edge in $T$ is to the right of $\ell'$.
Thus, by Lemma~\ref{lem:composition}, there is a way to compose $T$ and $T'$ to create an \e that does not use $j$ or some other job with earlier deadline than $j$ and release time before the first run time of $j$ in $\ell_1$.

We additionally note that there must be a $(j,\ell_1)$ active work barrier to the right of $\ell_1$.
Otherwise, we could use the active \e from such a $j'$ as part of a path of \es, which would be higher priority than $T$.
\qed
\end{proof}

\begin{lemma}
\label{lem:jworkbarrierPersistence}
If $t'$ is a $(j,\ell)$ (active) work barrier at $\gamma$ created from $j$ moving work right from $\ell$, and there is still some work of $j$ to the left of $\ell$, then for any $\gamma' > \gamma$ before $j$ begins moving work left, the algorithm takes no \e crossing $t'$ at $\gamma'$.
\end{lemma}

\begin{proof}
We prove the lemma for work barriers first.
Let $C_l$ be the last critical interval that is reachable by some job $j'$ with release after $j$ can be first run in $\ell$, and deadline before $j$.
It is clear that $j'$ must be at the same speed level as $j$ for this to be a problem.
There are two cases: either the work barrier could be removed by $C_l$ merging with another critical interval, or work from a job from beyond the work barrier enters $C_l$.
\begin{itemize}
\item \textbf{\boldmath$C_l$ merges with another critical interval.}
Let $T$ be the \e at $\gamma$ causing the work barrier.
We first show that, at $\gamma$, if $C_l$ can merge with another critical interval $C_r$ that would give $j'$ an \e over the work barrier to the destination of $j$ in the \e causing the work barrier, then either $C_l$ is at an upper speed level, or $C_r$ is at a lower speed level.
If not, $C_r$ could be the source of an \e with destination that of $T$, and source to the left of $T$, making which would be higher priority than $T$.
Additionally, we can obtain a series of \es that would be higher priority than $T$, depending on two cases
\begin{itemize}
\item \textbf{\boldmath Case 1: $j'$ is merged with $j$ in $\ell$.} In this case, $T$ could use $j'$ instead of $j$ and end in $C_l$.
\item \textbf{\boldmath Case 2: $j'$ is not merged with $j$ in $\ell$.}
Then $j$ is not at an upper speed level in $\ell$, and $j'$ is not at a lower speed level in $\ell$.
Thus, an active \e exists taking $j'$ to $C_l$, and $T$ could end in $\ell$.
\end{itemize}

If $C_l$ decreases due to $j'$ leaving over an edge $e$, then there is a $j'$-isolated area at the destination of $e$, and all \es must go through this depletion interval, and thus an \e crossing the work barrier would have to leave out some point other than the destination of $e$, contradicting that the algorithm took that \e (see Lemma~\ref{lem:edgeRemovalEvents}).
Thus it must be that $C_l$ is not at an upper speed level, and is not decreasing, and the critical interval it merges with is increasing.
The only remaining possibility is that $C_r$ was at a lower speed level at $\gamma$.
Note that any job $\tilde j$ in $C_r$ that could be used when $C_r$ and $C_l$ merge must have deadline after that of $j$, or be released after $\ell$, as otherwise the \e uses a job with deadline before that of $j$ that was alive in $\ell$ when $j$ could be first run there, or it contradicts the location of the work barrier at $\gamma$.
Note also that $C_r$ increasing cannot be due to the addition from the left of $j$ or any other job with release time before the first time $j$ can run in $\ell$ and deadline before that of $j$, as this would imply that $j'$ could be taken into $C_l$ instead, contradicting the \e taking the other job was used.
Similarly, $C_r$ increasing cannot be due to the addition from the right of $j$ or any other job with release time before the first time $j$ can run in $\ell$ and deadline before that of $j$, as this would create a $j$-isolated area, and no \e would enter it from outside (see Lemma~\ref{lem:edgeRemovalEvents}).

If $C_r$ increasing is due to some job $\tilde j$ with deadline before that of $j$, it must be coming from the right, as otherwise there would be an active \e from $j'$ using this job already.
The source of this \e must be to the left of $\ell$, or to the right of $C_l$.
In the first case, we could construct a higher priority \e with same source and destination ending in $C_l$.
In the second case, there is an isolated area at the right endpoint of $C_l$, and thus the only way to move work out of the isolated area is through $C_l$, so no \e taken would cross the work barrier.

Now suppose $C_r$ increasing is due to some job $\tilde j$ with deadline after that of $j$.
First assume that this \e taking $\tilde j$ is a right \e.
Note that $C_r$ must be in the destination of $j$ in $T$.
However, by the fact that $C_r$ was not increasing at $\gamma$, the destination of $j$ was not the destination of $T$, implying there is some work barrier to the right of this destination.
However, by Lemma~\ref{lem:workbarrierPersistence}, this work barrier could not have disappeared, contradicting that a right \e was being taken to this destination of $j$.
On the other hand, this \e is a left \e, there is an isolated area at the right endpoint of $C_l$, and thus the only way to move work out of the isolated area is through $C_l$, so no \e taken would cross the work barrier.

\item \textbf{\boldmath Work from some job beyond the work barrier enters $C_l$.}
Let this job be $\tilde j$.
First note that $\tilde j$ must be entering $C_l$ from the right, since it must be a higher priority job than $j$ as it's running between two times when $j$ is run, and must be released before the first time $j$ can run in $\ell$, or else it would not be able to go beyond the work barrier, but if so then the work barrier definition is not concerned with \es involving such jobs.
Thus, because the edge taking $\tilde j$ is a left edge, there is a $\tilde j$-isolated $I$ area starting at $C_l$.
Any \e using $j'$ before $\tilde j$ would contradict the property that no \es enter the isolated area (see Lemma~\ref{lem:edgeRemovalEvents}), and thus the algorithm never takes such an \e.
\end{itemize}

The active work barrier persists via an argument identical to that of Lemma~\ref{lem:workbarrierPersistence}, as we can again show there is an isolated area that ends at the work barrier.
\qed
\end{proof}

Finally, we bound the number of edge removal events.
We will need one Observation regarding the algorithm's avoidance of cycles.

\begin{observation}\label{obs:cycles}
The algorithm will never take an $\epsilon-$trasfer $(l_a,j_a)_{a = 0}^{s}$ such that for $a_1 \neq a_2$, $l_{a_1} = l_{a_2}$.
Intuitively this tells us the algorithm will never use an \e with a cycle.
\end{observation}

With this, we now get the following bound on the number of edge removals.

\begin{lemma}
\label{lem:edgeRemovalEvents}
The number of edge removal events between cuts is at most $O(n^3)$.
\end{lemma}

\begin{proof}
The high level idea of the proof is to show that for each job there are two phases of the algorithm between cuts.
The first phase involves moving work from this job left to right and the second phase involves moving work right to left.
To show this, we demonstrate that whenever a job moves work from left to right there is a work barrier that persists over time.
With this work barrier, it can be seen that this job will never move work to the right again.
With this in hand, we can show that the number of edge removals for each phase is polynomially bounded.
We now formalize this below.

Let $j$ be an arbitrary job and consider the first time there is an \e $T = (l_a,j_a)_{a = 0}^{s}$ chosen by the algorithm such that for some $a$, $j_{a'} = j$ and $l_{a' - 1} > l_{a'}$.
That is work from $j$ is moved right to left.
Let $t_1$ be the time that $j$ is run in $l_{a'}$.
We say that $A_j = [t_1,t_2]$ is a $\emph{j-isolated area}$ if $t_2$ is the minimum depletion point such that for every job $j'$ run after the first depletion interval in $A$, say $I_{A_j}$, either $d_{j'} \leq t_2$ and $r_{j'} \geq r_j$, or everything reachable by $j'$ is currently at an upper speed level.
Equivalently, this says that for any job inside the critical interval, the only way of moving work out is through the leftmost depletion interval.

The first step is to show that initially such a $\emph{j-isolated area}$ exists.
Assume by contradiction there is some \e that leaves $I_{A_j}$ through a depletion interval that is not the leftmost depletion interval.
There are two cases to consider.
Whether the edge leaving is a left going edge or a right going edge.

\begin{itemize}
\item \textbf{Case: Left going edge}
There are  two sub cases, depending on whether the source is inside the left edge taken by $j$ or whether the source is outside.
In both cases, we argue that you can form a higher priority \e.

Assume there is an active \e T' such that the source of T' is under the left edge chosen by $j$, that is,  $I_{a'} \leq I_{T'_s} \leq I_{a' - 1}$, and the destination $I_{T'_d}$ is to the left of $t_1$ and further $T'$ does not leave $I_{A_j}$ through the leftmost depletion interval.
We need to argue that you can combine part of $T'$ with part of the original \e $T$ to get a new \e $T''$ that is longer than the one chosen, contradicting our choice of $T$ as the longest left-going \e.
Specifically, take the original \e until we get to the edge that $T$ uses to leave $I_{A_j}$ and take this instead.
By Lemma~\ref{lem:composition} we can combine these to form $T''$  To argue that the algorithm would have chosen $T''$ instead of $T$ all that remains is to argue that the destination of $T''$ was a sink at the time $T$ was chosen.
This is a direct consequence of Lemma~\ref{lem:crossing}, that the algorithm does not choose crossing edges.

In the second sub case we can use similar approach here, the only difference being that we may need to combine several new \es to reach the same contradiction.

\item \textbf{Case: Right going edge}
We essentially just need to prove the lemma that it is not true that for every depletion point to the right that we can move work over that depletion point with a right going \e.
Equivalently, there exists at least one depletion point to the right of the source of $j$ that has no active right going \es over it.
If there were no such depletion point then combining this with Lemma~\ref{lem:crossing} would give us a sequence of right going \es that can be connected to the right-most depletion interval.
Again this would contradict our choice of $T$ as the algorithm preferences right going before left going \es.

\end{itemize}

The next step is to show that the properties of $I_{A_j}$ persist over time.
Namely, we need to show that the only way to exit $I_{A_j}$ is through the left most depletion interval.
There are two things we need to verify.
First, that no new work is placed in $I_{A_j}$, and second, that any critical interval reachable from a job in $I_{A_j}$ remains at an upper speed level.
To show the first claim, we consider two cases.
In the first case, assume that there is some \e chosen by the algorithm with source outside of $I_{A_j}$ and destination inside $I_{A_j}$.
This implies that when this \e is chosen, there is a path from the destination to the rightmost depletion interval.
However this results in either a crossing edge or a cycle, both of which the algorithm forbids.

In the second case assume there is some \e that puts work of some job $j'$ into $I_{A_j}$ but neither the source nor the destination are contained in $I_{A_j}$.
There are two sub cases to consider.
If the critical interval $j'$ is being placed into is at a lower speed level, then note that since middle pieces of \es do not change speeds this will not violate any property of $I_{A_j}$.
In the case where $j'$ is not at a lower speed level, this implies that taking the portion of the \e that starts in this critical interval and then leaves $I_{A_j}$ is an active \e, contradicting that all critical intervals currently reachable from $I_{A_j}$ are at  upper speed levels.

Lastly, we show that any critical interval reachable from a job in $I_{A_j}$ remains at an upper speed level.
Assume by contradiction that some the algorithm chooses some \e with a job $j'$ as the source such that $j'$ is part of a reachable critical interval from $I_{A_j}$ (Indeed this is the only way for a critical interval at an upper speed level to decrease).
Note that the destination of this \e is not an upper speed level.
However then by Lemma~\ref{lem:composition} we can combine this \e with an \e emanating from $I_{A_j}$ contradicting that all critical intervals from $I_{A_j}$ are currently at upper speed levels.

The last step is to show that as a result of the existence of the $j$-isolated interval, $j$ can no longer move work to the right.
Indeed assume by contradiction that at some future time point, $j$ is involved in an \e $(l_a,j_a)_{a = 0}^{s}$ where $j's$ work is moved from $l_{a - 1}$ to $l_a$ for some $a$ such that $l_{a - 1} < l_a$.
There are three cases to consider.

First consider the case when $l_{a - 1} < t_1$ and $l_a > t_1$.
Note by the definition of the $A$, the only way to move work out of $A$ is using the left most depletion interval $I_{A_j}$.
Further, by weak EDF, there can be no edge with source greater than $l_{a - 1}$ and destination greater than $l_a$ or less than $l_{a - 1}$.
These together tell us that at some point the algorithm chooses another \e with source $l_{a - 1}$, contradicting Observation~\ref{obs:cycles}.

In the second case assume both $l_{a - 1}$ and $l_a$ are not contained in $A$.
By definition of an active \e we know that $j_s$ is not at an upper speed level in $l_s$.
However we also know that $j_s$ is reachable by all pieces of $j$, at least one of which is contained in $A_j$.
This contradicts the second property of $A_j$, namely all jobs reachable from inside $A_j$ are at upper speed levels.

For the last case, assume that both $l_{a - 1}$ and $l_a$ are contained in $A$.
Similar to the first case, since we can only move work out of $A$ from the leftmost depletion interval, at some point the algorithm must choose another \e emanating from $l_{a - 1}$ again contradicting Observation~\ref{obs:cycles}.

Now that we have established a left and right phase for job $j$, we can show that there are only polynomially man edge removals for job $j$.

Suppose a job $j$ is moving out of a depletion interval along a right edge of an \e, and the work of $j$ is completely removed from the depletion interval.
Then $j$ never reenters the depletion interval by a right edge of an \e.

We show that any edge that could return the work of $j$ to the depletion interval $\ell$ must end to the left of a work barrier, which will contradict that the algorithm chose such an \e.

Let $\gamma$ be the event when the work of $j$ is completely removed from the depletion interval, and $T$ be the \e that removed $j$.
Let $\ell_2$ be the destination of $j$ in $T$.
Suppose that at some later event, the algorithm takes another \e $T'$ that moves work from $j$ back into $\ell$ from the left of $\ell$.
Our goal is to show that the path of \es from the destination of $T'$ must cross a work barrier, and any such \e that can do this would cause a weak EDF violation.

We first show that $T'$ must be a right \e.
To see this, first note that the edge $e$ taking $j$ in $T'$ is right, so it cannot come after a left edge containing it.
Additionally, $e$ cannot come after a series of smaller left edges starting from the right of $\ell$, as $j$ must be runnable in the entirety of $\ell$, so whatever edge of $T'$ that went left past $\ell$ must have been able to stop in $\ell$, contradicting the algorithm took $T'$.
A left edge going from the right of $\ell$ to the left of the right endpoint of $e$ is not possible, since, by Lemma~\ref{lem:jworkbarrierPersistence}, from the event when $T$ was taken, there is a $(j,\ell)$ work barrier somewhere before $\ell_2$, so any job taken would have to have deadline before $j$ and been released before $j$ can be first run in $\ell$, causing a weak EDF violation if this job were taken and $j$ were moved into $\ell$.
A series of left edges going from $\ell$ to the left of $\ell_2$ is also not possible, as either the edge crossing $\ell_2$ would cause a weak EDF violation, or the algorithm could have taken a more minimal \e not including $j$, contradicting that it took $T'$.

By Lemma~\ref{lem:jworkbarrier}, any destination of $T'$ must be to the left of a $(j,\ell)$ active work barrier, which persisted by Lemma~\ref{lem:jworkbarrierPersistence} from the one from the $j$ edge in $T$, as any job that could cross the work barrier would have to have deadline before $j$ and been released before $j$ can be first run in $\ell$, causing a weak EDF violation if this job were taken and $j$ were moved into $\ell$.
Thus, $T'$ must go to a sink that is the source of a shorter left \e, a contradiction that $T'$ was chosen, or the path of \es from $T'$ must use a job that could cross the work barrier, and thus has to have deadline before $j$ and been released before $j$ can be first run in $\ell$, causing a weak EDF violation by $j$ being moved into $\ell$, contradicting the algorithm took $T'$
\qed
\end{proof}
Combining all of these, we now give a Theorem bounding the total runtime of our algorithm.
\begin{theorem}\label{thm:algo_runtime}
The runtime of Algorithm~\ref{mainalgo} is $O(n^9k)$.
\end{theorem}
\begin{proof}
We first calculate an upper bound on the total times the algorithm stops for some event.
Multiplying this by the time spent in between events will give us our final bound.
Note that from Lemma~\ref{lem:runtime_cut_fixing} there are at most $O(kn^2)$ cut events.
Between two cut events there at most $O(n)$ depletion point events and $O(n)$ speed level events.
Finally, there are at most $O(n^3)$ edge removal events between any two other events.
Combining this gives us that
line $5$ of Algorithm~\ref{mainalgo} will be executed at most $O(kn^6)$ times.

To complete the analysis we need to bound lines $6-8$, line $9$, line $11$, and line $14$.
Consider first lines $6-8$.
Note that there are at most $O(n^2)$ unique job depletion interval pairs, and line $7$ is done in time $O(1)$.
Line $8$ also takes time $O(1)$ giving us total time $O(n^2)$.

To bound line $9$, the UpdateGraph procedure, recall that we recalculate both the schedule and the graph.
To calculate the schedule note there are $n$ depletion intervals each with at most $n$ jobs.
Since the YDS algorithm runs in time $O(n\log^2 n)$, our total time to compute the schedule is $O(n^2 \log^2 n)$.
For computing the subgraph of the distribution graph, recall that for every depletion interval our algorithm uses a depth-first search on $n$ vertices and up to $n^2$ edges, giving a total runtime of $nO(|E|) = O(n^3)$.

To bound line $11$, note that checking each depletion point for possible removal takes time $O(n)$ and fixing speed levels takes time at most $O(n^2)$, since there are at most $n$ different speed levels for each job.

Finally to bound line $14$ note to calculate the rates such that depletion points remain can be done in $O(n^2)$, since there at most $O(n)$ atomic intervals inside each depletion interval.
Calculating the next event for edge removals takes time $O(n^2)$.
For speed level events, since we must consider the possibility that jobs merge along the way, there are $O(n^2)$ calculations for each depletion interval, and therefore takes a total of $O(n^3)$.

Combining this, we see that the runtime is $O(kn^6(n^2 \log^2 n + n^3 + n^2 + n^3)) = O(kn^9)$.
\qed
\end{proof}


\section{Algorithm Correctness}\label{app:alg_correctness}
In this section we demonstrate our algorithm correctly finds the optimal schedule.
At a high level, in we show that at all steps of the algorithm we maintain optimality conditions 1-3 and that when the algorithm can no longer make progress it satisfies the last optimality condition.
We first consider movement of work by the algorithm, and its handling of speed level and edge removal events.
We then consider cut events.

\subsection{Non-Cut Events}
We now show that moving work and handling non-cut events does not violate any of the optimality conditions maintained throughout the algorithm's execution.

\begin{theorem}\label{thm:maintaining_some_optimality_properties}
The algorithm maintains Properties~\ref{thm:optimality:a} to~\ref{thm:optimality:c} of Theorem~\ref{thm:optimality} when moving work and handling non-cut events.
\end{theorem}
\begin{proof}
Clearly, when the recharge rate is initially set from the YDS schedule of the instance, Properties~\ref{thm:optimality:a} and~\ref{thm:optimality:b} are maintained.
Property~\ref{thm:optimality:c} follows from the YDS property of the schedule:
First, every job is assigned the single speed level corresponding to the speeds at which it is run when it is run;
Second, if there were a time at which a job $j$ is alive but some other job $j'$ was being run, with the speed level of $j'$ less than the speed level of $j$, this would contradict the YDS property as $j'$ would have to be part of the critical interval of $j$ and thus have the same speed level.

\textbf{Moving Work.}
Since the step of moving work does not change the speed levels of jobs, and the process stops when critical intervals hit a new discrete speed (i.e., an edge becomes inactive) or a job no longer has work to move, moving work cannot violate Property~\ref{thm:optimality:c}.
Additionally, work is moved in a work-preserving manner for each job, and the process stops if at some new time the total amount of energy available becomes zero (i.e., a depletion point appears), so Property~\ref{thm:optimality:a} is maintained.
Property~\ref{thm:optimality:b} is maintained by the fact that we have a YDS schedule between critical intervals.
By the fact that at non-cut events, only new \es are chosen, the handling of these events does not violate Properties~\ref{thm:optimality:a} to~\ref{thm:optimality:c}.
\qed
\end{proof}

\subsection{Maintaining the SLR \& Reaching Optimality}\label{sec:handling_cuts}
Consider a situation when there is a depletion interval $\ell$ for which there is no path to $L+1$ in the distribution graph.
This means we are unable to move workload out of this depletion interval – and, thus, cannot lower the recharge rate – without violating the SLR or other optimality conditions.
The following lemmas take a closer look at such situations.
In particular, we show that we either can fix the speed levels, adapt the set of depletion intervals, or have found an optimal solution.

\paragraph{Fixing Speed Levels}
We start with the most intuitive reason for not being able to make progress: there are jobs that could still transfer work to the rightmost depletion interval (possibly taking several \es), but the SLR requirement renders any such path inactive.
For a single \e, we can easily change the speed levels such that using this \e will not violate Properties~\ref{def:slr:a} or~\ref{def:slr:b} of the SLR.
However, to handle~\ref{def:slr:c} and~\ref{def:slr:d}, we have to take care to adapt the speed levels of certain jobs in a compatible way.
The following lemma takes care of that.
\begin{lemma}\label{lem:fix_normal_cuts}
Assume there is a depletion interval $\ell_0$ without a path of active \es to $L+1$.
Furthermore assume there is a path to $L+1$ utilizing at least one inactive \e.
Then we can fix the speed levels and increase the set of nodes reachable from $\ell_0$.
\end{lemma}
\begin{proof}
Let $\ell_{\min}$ and $\ell_{\max}$ denote the minimal and maximal depletion intervals reachable from $\ell_0$ via a path of active \es.
Obviously, no depletion interval in $\set{\ell_{\min},\ell_{\min}+1,\dots,\ell_{\max}}$ can reach $L+1$ via a path of active \es.
In fact, by definition, no active \e leaves the union of these depletion intervals.
Now fix an arbitrary (inactive) \e $T$ leaving $U\coloneqq\bigcup_{\ell=\ell_{\min}}^{\ell_{\min}}I_l$.
Let $j$ denote the job used to move work out of $U$ in $T$ and define the job set $J_j$ of jobs that are part of some (active or inactive) \e that also uses $j$.
Note that either
\begin{enumerate*}
\item any job $j'\in J_j$ is at a lower speed level in $U$ or
\item any job $j'\in J_j$ is at a higher speed level outside of $U$.
\end{enumerate*}
If that were not true, we get two (inactive) \es leaving $U$, one that can be lowered at its source (but is inactive because its destination is at an upper speed level) and one that can be increased at its destination (but is inactive because its source is at a lower speed level).
Since both of these use $j_a$, we can concatenate them to get an active \e leaving $U$, contradicting our assumption.
We now can simultaneously fix the speed levels of alls jobs $j'\in J_j$ by either
\begin{enumerate*}
\item decreasing their speed levels $\SL(j',\ell)$ for all $\ell\in\set{\ell_{\min},\ell_{\min}+1,\dots,\ell_{\max}}$ or
\item increasing their speed levels $\SL(j',\ell)$ for all $\ell\notin\set{\ell_{\min},\ell_{\min}+1,\dots,\ell_{\max}}$.
\end{enumerate*}
We iterate this until there are no more inactive \es leaving $U$.

It is easy to check that this procedure maintains Properties~\ref{def:slr:a},~\ref{def:slr:b}, and~\ref{def:slr:d} of the SLR.
For Property~\ref{def:slr:c}, note that the described procedure
\begin{enumerate*}
\item does not change the speed level difference between any two depletion intervals that are both inside or both outside of $U$ and
\item any job that is active in- and outside of $U$ will be part of some \e considered by the procedure, and thus its differences considered in Property~\ref{def:slr:c} of the SLR decrease by exactly one at the left border of $U$ and increase by exactly one at the right border of $U$.
\end{enumerate*}
After fixing the speed levels, the set of nodes reachable from $\ell_0$ has increased beyond $U$.
This proves the lemma's statement.
\qed
\end{proof}

\paragraph{Merging Depletion Intervals}
Lemma~\ref{lem:fix_normal_cuts} shows how to continue to make progress in some situations.
However there are still cases that don't allow progress with the current distribution graph but are not covered by that lemma.
Also, there is a subtlety in the lemma's proof when we decrease speed levels.
Consider a job $j$ whose speed level in $\ell$ gets decreased.
By Property~\ref{def:slr:c}, we have to make sure that, for any $\ell'<\ell$, the difference $\delta_{\ell,\ell'}\coloneqq\SL(j,\ell)-\SL(j,\ell')$ remains non-negative (note that, by the same property, this value does not depend on $j$).
To guarantee this, we merge any two depletion intervals as soon as the value $\delta_{\ell,\ell'}$ becomes zero.

\paragraph{Reaching Optimality}
It remains to show that if (after repairing speed levels and merging depletion intervals) we still cannot make progress, we actually found an optimal solution.
We prove this in the following lemma.
\begin{lemma}\label{lem:reached_optimality}
Assume there is a depletion interval $\ell_0$ for which there is no path of (active or inactive) \es to $L+1$.
Then the current solution is optimal.
\end{lemma}
\begin{proof}
We show the optimality of the current solution by showing that Property~\ref{thm:optimality:d} of Theorem~\ref{thm:optimality} holds.
Together with the fact that our algorithm maintains Properties~\ref{thm:optimality:a} to~\ref{thm:optimality:c} of this theorem (see Theorem~\ref{thm:maintaining_some_optimality_properties}), this implies that the current solution is optimal.

To see that Property~\ref{thm:optimality:d} holds, first notice that any job processed at a time $t\in I_{\ell}$ with deadline $d_j\in I_{\ell'}$ ($\ell'>\ell$) implies a (possibly inactive) \e of length one from $\ell$ to $\ell'$ (by just moving work of $j$).
Now let $\ell_0$ denote the rightmost depletion interval without a path of \es to $L+1$.
Assume Property~\ref{thm:optimality:d} is not true, so there is a job $j$ with deadline $d_j>\tau_{\ell_0}$ and $j$ is processed before $\tau_{\ell_0}$.
Let $\ell_1\leq\ell_0$ denote the depletion interval in which $j$ is processed the first time and $\ell_2>\ell_0$ the depletion interval that contains its deadline $d_j$.
If $\ell_1=\ell_0$, we're done: As noticed above, any such job would imply a (possibly inactive) \e of length one to $\ell_2$.
But then, as $j$ was chosen maximal, there is a path from $\ell_2$ to $L+1$.
Together, these $\epsilon$ transfers form a path from $\ell_0$ to $L+1$, a contradiction.
So consider the case $\ell_1<\ell_0$.
Note that, without loss of generality, we can assume there is a job $j'$ that is processed in $\ell_0$ and has release time $r_{j'}<\tau_{\ell_0-1}$.
If there is no such job, we could simply merge the depletion intervals $\ell_0$ and $\ell_0-1$.
Similar to the argument above, this job $j'$ gives us a (possibly inactive) \e of length one from $\ell_0$ to the left.
By iterating this argument, we get a path of \es from $\ell_0$ to $\ell_1$ and, thus, to $\ell_2$.
As before, this yields a contradiction.
\qed
\end{proof}

\end{document}